\documentclass{article}

\usepackage{rotating}
\usepackage{floatpag}
\rotfloatpagestyle{empty}

\usepackage{hyperref}

\usepackage{amsmath,amsthm,amssymb}
\usepackage{xspace,enumerate,color,epsfig}
\usepackage{graphicx}
\graphicspath{{.}{./figures/}}
\usepackage{marvosym}

\usepackage{tikzfig}
\usepackage{stmaryrd}
\usepackage{docmute}
\usepackage{keycommand}

\newcommand{\emptydiag}{\,\tikz{\node[style=empty diagram] (x) {};}\,}

\newcommand{\boxmap}[1]{\,\tikz{\node[style=small box] (x) {$#1$};\draw(0,-1)--(x)--(0,1);}\,}

\newkeycommand{\pointmap}[style=point][1]{\,\tikz{\node[style=\commandkey{style}] (x) at (0,-0.3) {$#1$};\node [style=none] (3) at (0, 1.0) {};\node [style=none] (3) at (0, -1.0) {};\draw (x)--(0,0.7);}\,}

\newkeycommand{\typedkpoint}[style=kpoint,edgestyle=][2]{%
\,\begin{tikzpicture}
  \begin{pgfonlayer}{nodelayer}
    \node [style=none] (0) at (0, 1) {};
    \node [style=\commandkey{style}] (1) at (0, -0.75) {$#2$};
    \node [style=right label] (2) at (0.25, 0.5) {$#1$};
  \end{pgfonlayer}
  \begin{pgfonlayer}{edgelayer}
    \draw [\commandkey{edgestyle}] (1) to (0.center);
  \end{pgfonlayer}
\end{tikzpicture}\,}

\newkeycommand{\typedkpointdag}[style=kpoint adjoint,edgestyle=][2]{%
\,\begin{tikzpicture}
  \begin{pgfonlayer}{nodelayer}
    \node [style=none] (0) at (0, -1) {};
    \node [style=\commandkey{style}] (1) at (0, 0.75) {$#2$};
    \node [style=right label] (2) at (0.25, -0.5) {$#1$};
  \end{pgfonlayer}
  \begin{pgfonlayer}{edgelayer}
    \draw [\commandkey{edgestyle}] (0.center) to (1);
  \end{pgfonlayer}
\end{tikzpicture}\,}

\newcommand{\bistateadj}[1]{\,\begin{tikzpicture}[yshift=-3mm]
  \begin{pgfonlayer}{nodelayer}
    \node [style=kpointadj, minimum width=1 cm, inner sep=2pt] (0) at (0, 0.5) {$\psi$};
    \node [style=none] (1) at (-0.75, 0.25) {};
    \node [style=none] (2) at (0.75, 0.25) {};
    \node [style=none] (3) at (-0.75, -0.5) {};
    \node [style=none] (4) at (0.75, -0.5) {};
  \end{pgfonlayer}
  \begin{pgfonlayer}{edgelayer}
    \draw (1.center) to (3.center);
    \draw (2.center) to (4.center);
  \end{pgfonlayer}
\end{tikzpicture}\,}

\newkeycommand{\pointketbra}[style1=copoint,style2=point][2]{\,%
\begin{tikzpicture}
  \begin{pgfonlayer}{nodelayer}
    \node [style=none] (0) at (0, -1.5) {};
    \node [style=\commandkey{style1}] (1) at (0, -0.75) {$#1$};
    \node [style=\commandkey{style2}] (2) at (0, 0.75) {$#2$};
    \node [style=none] (3) at (0, 1.5) {};
  \end{pgfonlayer}
  \begin{pgfonlayer}{edgelayer}
    \draw (0.center) to (1);
    \draw (2) to (3.center);
  \end{pgfonlayer}
\end{tikzpicture}\,}

\newkeycommand{\twocopointketbra}[style1=copoint,style2=copoint,style3=point][3]{\,%
\begin{tikzpicture}
  \begin{pgfonlayer}{nodelayer}
    \node [style=none] (0) at (-0.75, -1.5) {};
    \node [style=none] (1) at (0.75, -1.5) {};
    \node [style=\commandkey{style1}] (2) at (-0.75, -0.75) {$#1$};
    \node [style=\commandkey{style2}] (3) at (0.75, -0.75) {$#2$};
    \node [style=\commandkey{style3}] (4) at (0, 0.75) {$#3$};
    \node [style=none] (5) at (0, 1.5) {};
  \end{pgfonlayer}
  \begin{pgfonlayer}{edgelayer}
    \draw (0.center) to (2);
    \draw (1.center) to (3);
    \draw (4) to (5.center);
  \end{pgfonlayer}
\end{tikzpicture}\,}

\newkeycommand{\twopointketbra}[style1=copoint,style2=point,style3=point][3]{\,%
\begin{tikzpicture}
  \begin{pgfonlayer}{nodelayer}
    \node [style=none] (0) at (0, -1.5) {};
    \node [style=none] (1) at (-0.75, 1.5) {};
    \node [style=\commandkey{style1}] (2) at (0, -0.75) {$#1$};
    \node [style=\commandkey{style2}] (3) at (-0.75, 0.75) {$#2$};
    \node [style=\commandkey{style3}] (4) at (0.75, 0.75) {$#3$};
    \node [style=none] (5) at (0.75, 1.5) {};
  \end{pgfonlayer}
  \begin{pgfonlayer}{edgelayer}
    \draw (0.center) to (2);
    \draw (1.center) to (3);
    \draw (4) to (5.center);
  \end{pgfonlayer}
\end{tikzpicture}\,}

\newcommand{\shortidwire}{\,%
\begin{tikzpicture}
  \begin{pgfonlayer}{nodelayer}
    \node [style=none] (0) at (0, -0.8) {};
    \node [style=none] (3) at (0, 0.8) {};
  \end{pgfonlayer}
  \begin{pgfonlayer}{edgelayer}
    \draw (0.center) to (3.center);
  \end{pgfonlayer}
\end{tikzpicture}\,}

\newkeycommand{\pointbraket}[style1=point,style2=copoint][2]{\,%
\begin{tikzpicture}
  \begin{pgfonlayer}{nodelayer}
    \node [style=\commandkey{style1}] (0) at (0, -0.5) {$#1$};
    \node [style=\commandkey{style2}] (1) at (0, 0.5) {$#2$};
  \end{pgfonlayer}
  \begin{pgfonlayer}{edgelayer}
    \draw (0) to (1);
  \end{pgfonlayer}
\end{tikzpicture}\,}

\newkeycommand{\kpointbraket}[style1=kpoint,style2=kpoint adjoint][2]{\,%
\begin{tikzpicture}
  \begin{pgfonlayer}{nodelayer}
    \node [style=\commandkey{style1}] (0) at (0, -0.75) {$#1$};
    \node [style=\commandkey{style2}] (1) at (0, 0.75) {$#2$};
  \end{pgfonlayer}
  \begin{pgfonlayer}{edgelayer}
    \draw (0) to (1);
  \end{pgfonlayer}
\end{tikzpicture}\,}

\newkeycommand{\kpointmap}[style1=kpoint,style2=map][2]{\,%
\begin{tikzpicture}
  \begin{pgfonlayer}{nodelayer}
    \node [style=\commandkey{style1}] (0) at (0, -1.1) {$#1$};
    \node [style=\commandkey{style2}] (1) at (0, 0.2) {$#2$};
    \node [style=none] (2) at (0, 1.5) {};
  \end{pgfonlayer}
  \begin{pgfonlayer}{edgelayer}
    \draw (0) to (1);
    \draw (1) to (2.center);
  \end{pgfonlayer}
\end{tikzpicture}%
\,}

\newkeycommand{\sandwichtwo}[style1=point,style2=map,style3=map,style4=copoint][4]{\,%
\begin{tikzpicture}
  \begin{pgfonlayer}{nodelayer}
    \node [style=\commandkey{style2}] (0) at (0, -0.75) {$#2$};
    \node [style=\commandkey{style3}] (1) at (0, 0.75) {$#3$};
    \node [style=\commandkey{style1}] (2) at (0, -2) {$#1$};
    \node [style=\commandkey{style4}] (3) at (0, 2) {$#4$};
  \end{pgfonlayer}
  \begin{pgfonlayer}{edgelayer}
    \draw (2) to (0);
    \draw (0) to (1);
    \draw (1) to (3);
  \end{pgfonlayer}
\end{tikzpicture}\,}

\newkeycommand{\longbraket}[style1=point,style2=copoint][2]{\,%
\begin{tikzpicture}
  \begin{pgfonlayer}{nodelayer}
    \node [style=\commandkey{style1}] (0) at (0, -2) {$#1$};
    \node [style=\commandkey{style2}] (1) at (0, 2) {$#2$};
  \end{pgfonlayer}
  \begin{pgfonlayer}{edgelayer}
    \draw (0) to (1);
  \end{pgfonlayer}
\end{tikzpicture}\,}

\newkeycommand{\onb}[style=point]{\ensuremath{\left\{\tikz{\node[style=\commandkey{style}] (x) at (0,-0.3) {$j$};\draw (x)--(0,0.7);}\right\}}\xspace}

\newkeycommand{\copointmap}[style=copoint][1]{\,\tikz{\node[style=\commandkey{style}] (x) at (0,0.3) {$#1$};\draw (x)--(0,-0.7);}\,}

\tikzstyle{cdiag}=[matrix of math nodes, row sep=3em, column sep=3em, text height=1.5ex, text depth=0.25ex,inner sep=0.5em]
\tikzstyle{arrow above}=[transform canvas={yshift=0.5ex}]
\tikzstyle{arrow below}=[transform canvas={yshift=-0.5ex}]

\usetikzlibrary{shapes.multipart}

\tikzstyle{bwSpider}=[
       rectangle split,
       rectangle split parts=2,
       rectangle split part fill={black,white},
 minimum size=3.6 mm, inner sep=-2mm, draw=black,scale=0.5,rounded corners=0.8 mm
       ]
 \tikzstyle{wbSpider}=[
       rectangle split,
       rectangle split parts=2,
       rectangle split part fill={white,black},
 minimum size=3.6 mm, inner sep=-2mm, draw=black,scale=0.5,rounded corners=0.8 mm
       ]
\tikzstyle{cWire}=[densely dotted, thick]
\tikzstyle{env}=[copoint,regular polygon rotate=0,minimum width=0.2cm, fill=black]

\tikzstyle{probs}=[shape=semicircle,fill=white,draw=black,shape border rotate=180,minimum width=1.2cm]

\tikzstyle{every picture}=[baseline=-0.25em,scale=0.5]
\tikzstyle{dotpic}=[] 
\tikzstyle{diredges}=[every to/.style={diredge}]
\tikzstyle{math matrix}=[matrix of math nodes,left delimiter=(,right delimiter=),inner sep=2pt,column sep=1em,row sep=0.5em,nodes={inner sep=0pt},text height=1.5ex, text depth=0.25ex]

\tikzstyle{inline text}=[text height=1.5ex, text depth=0.25ex,yshift=0.5mm]
\tikzstyle{label}=[font=\footnotesize,text height=1.5ex, text depth=0.25ex,yshift=0.5mm]
\tikzstyle{left label}=[label,anchor=east,xshift=1.5mm]
\tikzstyle{right label}=[label,anchor=west,xshift=-1.5mm]

\tikzstyle{braceedge}=[decorate,decoration={brace,amplitude=2mm,raise=-1mm}]
\tikzstyle{small braceedge}=[decorate,decoration={brace,amplitude=1mm,raise=-1mm}]

\tikzstyle{doubled}=[line width=1.6pt] 
\tikzstyle{boldedge}=[doubled,shorten <=-0.17mm,shorten >=-0.17mm]
\tikzstyle{boldedgegray}=[doubled,gray,shorten <=-0.17mm,shorten >=-0.17mm]
\tikzstyle{singleedgegray}=[gray]

\tikzstyle{semidoubled}=[line width=1.4pt] 
\tikzstyle{semiboldedgegray}=[semidoubled,gray,shorten <=-0.17mm,shorten >=-0.17mm]

\tikzstyle{boxedge}=[semiboldedgegray]

\tikzstyle{boldedgedashed}=[very thick,dashed,shorten <=-0.17mm,shorten >=-0.17mm]
\tikzstyle{vboldedgedashed}=[doubled,dashed,shorten <=-0.17mm,shorten >=-0.17mm]
\tikzstyle{left hook arrow}=[left hook-latex]
\tikzstyle{right hook arrow}=[right hook-latex]
\tikzstyle{sembracket}=[line width=0.5pt,shorten <=-0.07mm,shorten >=-0.07mm]

\tikzstyle{causal edge}=[->,thick,gray]
\tikzstyle{causal nondir}=[thick,gray]
\tikzstyle{timeline}=[thick,gray, dashed]

\tikzstyle{cedge}=[<->,thick,gray!70!white]

\tikzstyle{empty diagram}=[draw=gray!40!white,dashed,shape=rectangle,minimum width=1cm,minimum height=1cm]
\tikzstyle{empty diagram small}=[draw=gray!50!white,dashed,shape=rectangle,minimum width=0.6cm,minimum height=0.5cm]

\tikzstyle{dot}=[inner sep=0mm,minimum width=2mm,minimum height=2mm,draw,shape=circle]

\tikzstyle{leak}=[white dot, shape=regular polygon, minimum size=3.3 mm, regular polygon sides=3, outer sep=-0.2mm, regular polygon rotate=270]
\tikzstyle{preleak}=[trapezium, trapezium angle=67.5, draw, inner sep=0.1pt, outer sep=0pt, minimum height=2mm, minimum width=4pt,rotate=270]
\tikzstyle{proj}=[white dot, shape=regular polygon, minimum size=3.3 mm, regular polygon sides=4, outer sep=-0.2mm]
\tikzstyle{Vleak}=[white dot, shape=regular polygon, minimum size=3.3 mm, regular polygon sides=3, outer sep=-0.2mm, regular polygon rotate=90]
\tikzstyle{dleak}=[white dot, line width=1.6pt, shape=regular polygon, minimum size=3.3 mm, regular polygon sides=3, outer sep=-0.2mm, regular polygon rotate=270]

\tikzstyle{Wsquare}=[white dot, shape=regular polygon, rounded corners=0.8 mm, minimum size=3.3 mm, regular polygon sides=3, outer sep=-0.2mm]
\tikzstyle{Wsquareadj}=[white dot, shape=regular polygon, rounded corners=0.8 mm, minimum size=3.3 mm, regular polygon sides=3, outer sep=-0.2mm, regular polygon rotate=180]
\tikzstyle{ddot}=[inner sep=0mm, doubled, minimum width=2.5mm,minimum height=2.5mm,draw,shape=circle]

\tikzstyle{black dot}=[dot,fill=black]
\tikzstyle{white dot}=[dot,fill=white,,text depth=-0.2mm]
\tikzstyle{white Wsquare}=[Wsquare,fill=gray,,text depth=-0.2mm]
\tikzstyle{white Wsquareadj}=[Wsquareadj,fill=white,,text depth=-0.2mm]
\tikzstyle{green dot}=[white dot] 
\tikzstyle{gray dot}=[dot,fill=gray!40!white,,text depth=-0.2mm]
\tikzstyle{red dot}=[gray dot] 

\tikzstyle{black ddot}=[ddot,fill=black]
\tikzstyle{white ddot}=[ddot,fill=white]
\tikzstyle{gray ddot}=[ddot,fill=gray!40!white]

\tikzstyle{gray edge}=[gray!60!white]

\tikzstyle{small dot}=[inner sep=0.5mm,minimum width=0pt,minimum height=0pt,draw,shape=circle]

\tikzstyle{small black dot}=[small dot,fill=black]
\tikzstyle{small white dot}=[small dot,fill=white]
\tikzstyle{small gray dot}=[small dot,fill=gray!40!white]

\tikzstyle{causal dot}=[inner sep=0.4mm,minimum width=0pt,minimum height=0pt,draw=white,shape=circle,fill=gray!40!white]

\tikzstyle{phase dimensions}=[minimum size=5mm,font=\footnotesize,rectangle,rounded corners=2.5mm,inner sep=0.2mm,outer sep=-2mm]
\tikzstyle{dphase dimensions}=[minimum size=5mm,font=\footnotesize,rectangle,rounded corners=2.5mm,inner sep=0.2mm,outer sep=-2mm]

\tikzstyle{white phase dot}=[dot,fill=white,phase dimensions]
\tikzstyle{white phase ddot}=[ddot,fill=white,dphase dimensions]

\tikzstyle{white rect ddot}=[draw=black,fill=white,doubled,minimum size=5mm,font=\footnotesize,rectangle,rounded corners=2.5mm,inner sep=0.2mm]
\tikzstyle{gray rect ddot}=[draw=black,fill=gray!40!white,doubled,minimum size=6mm,font=\footnotesize,rectangle,rounded corners=3mm]

\tikzstyle{gray phase dot}=[dot,fill=gray!40!white,phase dimensions]
\tikzstyle{gray phase ddot}=[ddot,fill=gray!40!white,dphase dimensions]
\tikzstyle{grey phase dot}=[gray phase dot]
\tikzstyle{grey phase ddot}=[gray phase ddot]

\tikzstyle{small phase dimensions}=[minimum size=4mm,font=\tiny,rectangle,rounded corners=2mm,inner sep=0.2mm,outer sep=-2mm]
\tikzstyle{small dphase dimensions}=[minimum size=4mm,font=\tiny,rectangle,rounded corners=2mm,inner sep=0.2mm,outer sep=-2mm]

\tikzstyle{small gray phase dot}=[dot,fill=gray!40!white,small phase dimensions]
\tikzstyle{small gray phase ddot}=[ddot,fill=gray!40!white,small dphase dimensions]

\tikzstyle{small map}=[draw,shape=rectangle,minimum height=4mm,minimum width=4mm,fill=white]

\tikzstyle{cnot}=[fill=white,shape=circle,inner sep=-1.4pt]

\tikzstyle{asym hadamard}=[fill=white,draw,shape=NEbox,inner sep=0.6mm,font=\footnotesize,minimum height=4mm]
\tikzstyle{asym hadamard conj}=[fill=white,draw,shape=NWbox,inner sep=0.6mm,font=\footnotesize,minimum height=4mm]
\tikzstyle{asym hadamard dag}=[fill=white,draw,shape=SEbox,inner sep=0.6mm,font=\footnotesize,minimum height=4mm]

\tikzstyle{hadamard}=[fill=white,draw,inner sep=0.6mm,font=\footnotesize,minimum height=4mm,minimum width=4mm]
\tikzstyle{small hadamard}=[fill=white,draw,inner sep=0.6mm,minimum height=1.5mm,minimum width=1.5mm]
\tikzstyle{small hadamard rotate}=[small hadamard,rotate=45]
\tikzstyle{dhadamard}=[hadamard,doubled]
\tikzstyle{small dhadamard}=[small hadamard,doubled]
\tikzstyle{small dhadamard rotate}=[small hadamard rotate,doubled]
\tikzstyle{antipode}=[white dot,inner sep=0.3mm,font=\footnotesize]

\tikzstyle{scalar}=[diamond,draw,inner sep=0.5pt,font=\small]
\tikzstyle{dscalar}=[diamond,doubled, draw,inner sep=0.5pt,font=\small]

\tikzstyle{small box}=[rectangle,inline text,fill=white,draw,minimum height=5mm,yshift=-0.5mm,minimum width=5mm,font=\small]
\tikzstyle{small gray box}=[small box,fill=gray!30]
\tikzstyle{medium box}=[rectangle,inline text,fill=white,draw,minimum height=5mm,yshift=-0.5mm,minimum width=10mm,font=\small]
\tikzstyle{square box}=[small box] 
\tikzstyle{medium gray box}=[small box,fill=gray!30]
\tikzstyle{semilarge box}=[rectangle,inline text,fill=white,draw,minimum height=5mm,yshift=-0.5mm,minimum width=12.5mm,font=\small]
\tikzstyle{large box}=[rectangle,inline text,fill=white,draw,minimum height=5mm,yshift=-0.5mm,minimum width=15mm,font=\small]
\tikzstyle{large gray box}=[small box,fill=gray!30]

\tikzstyle{Bayes box}=[rectangle,fill=black,draw, minimum height=3mm, minimum width=3mm]

\tikzstyle{gray square point}=[small box,fill=gray!50]

\tikzstyle{dphase box white}=[dhadamard]
\tikzstyle{dphase box gray}=[dhadamard,fill=gray!50!white]
\tikzstyle{phase box white}=[hadamard]
\tikzstyle{phase box gray}=[hadamard,fill=gray!50!white]

\tikzstyle{point}=[regular polygon,regular polygon sides=3,draw,scale=0.75,inner sep=-0.5pt,minimum width=9mm,fill=white,regular polygon rotate=180]
\tikzstyle{point nosep}=[regular polygon,regular polygon sides=3,draw,scale=0.75,inner sep=-2pt,minimum width=9mm,fill=white,regular polygon rotate=180]
\tikzstyle{copoint}=[regular polygon,regular polygon sides=3,draw,scale=0.75,inner sep=-0.5pt,minimum width=9mm,fill=white]
\tikzstyle{dpoint}=[point,doubled]
\tikzstyle{dcopoint}=[copoint,doubled]

\tikzstyle{pointgrow}=[shape=cornerpoint,kpoint common,scale=0.75,inner sep=3pt]
\tikzstyle{pointgrow dag}=[shape=cornercopoint,kpoint common,scale=0.75,inner sep=3pt]

\tikzstyle{wide copoint}=[fill=white,draw,shape=isosceles triangle,shape border rotate=90,isosceles triangle stretches=true,inner sep=0pt,minimum width=1.5cm,minimum height=6.12mm]
\tikzstyle{wide point}=[fill=white,draw,shape=isosceles triangle,shape border rotate=-90,isosceles triangle stretches=true,inner sep=0pt,minimum width=1.5cm,minimum height=6.12mm,yshift=-0.0mm]
\tikzstyle{wide point plus}=[fill=white,draw,shape=isosceles triangle,shape border rotate=-90,isosceles triangle stretches=true,inner sep=0pt,minimum width=1.74cm,minimum height=7mm,yshift=-0.0mm]

\tikzstyle{wide dpoint}=[fill=white,doubled,draw,shape=isosceles triangle,shape border rotate=-90,isosceles triangle stretches=true,inner sep=0pt,minimum width=1.5cm,minimum height=6.12mm,yshift=-0.0mm]

\tikzstyle{tinypoint}=[regular polygon,regular polygon sides=3,draw,scale=0.55,inner sep=-0.15pt,minimum width=6mm,fill=white,regular polygon rotate=180]

\tikzstyle{white point}=[point]
\tikzstyle{white dpoint}=[dpoint]
\tikzstyle{green point}=[white point] 
\tikzstyle{white copoint}=[copoint]
\tikzstyle{gray point}=[point,fill=gray!40!white]
\tikzstyle{gray dpoint}=[gray point,doubled]
\tikzstyle{red point}=[gray point] 
\tikzstyle{gray copoint}=[copoint,fill=gray!40!white]
\tikzstyle{gray dcopoint}=[gray copoint,doubled]

\tikzstyle{white point guide}=[regular polygon,regular polygon sides=3,font=\scriptsize,draw,scale=0.65,inner sep=-0.5pt,minimum width=9mm,fill=white,regular polygon rotate=180]

\tikzstyle{black point}=[point,fill=black,font=\color{white}]
\tikzstyle{black copoint}=[copoint,fill=black,font=\color{white}]

\tikzstyle{tiny gray point}=[tinypoint,fill=gray!40!white]

\tikzstyle{diredge}=[->]
\tikzstyle{ddiredge}=[<->]
\tikzstyle{rdiredge}=[<-]
\tikzstyle{thickdiredge}=[->, very thick]
\tikzstyle{pointer edge}=[->,very thick,gray]
\tikzstyle{pointer edge part}=[very thick,gray]
\tikzstyle{dashed edge}=[dashed]
\tikzstyle{thick dashed edge}=[very thick,dashed]
\tikzstyle{thick gray dashed edge}=[thick dashed edge,gray!40]
\tikzstyle{thick map edge}=[very thick,|->]

\makeatletter
\newcommand{\boxshape}[3]{%
\pgfdeclareshape{#1}{
\inheritsavedanchors[from=rectangle] 
\inheritanchorborder[from=rectangle]
\inheritanchor[from=rectangle]{center}
\inheritanchor[from=rectangle]{north}
\inheritanchor[from=rectangle]{south}
\inheritanchor[from=rectangle]{west}
\inheritanchor[from=rectangle]{east}
\backgroundpath{
\southwest \pgf@xa=\pgf@x \pgf@ya=\pgf@y
\northeast \pgf@xb=\pgf@x \pgf@yb=\pgf@y

\@tempdima=#2
\@tempdimb=#3

\pgfpathmoveto{\pgfpoint{\pgf@xa - 5pt + \@tempdima}{\pgf@ya}}
\pgfpathlineto{\pgfpoint{\pgf@xa - 5pt - \@tempdima}{\pgf@yb}}
\pgfpathlineto{\pgfpoint{\pgf@xb + 5pt + \@tempdimb}{\pgf@yb}}
\pgfpathlineto{\pgfpoint{\pgf@xb + 5pt - \@tempdimb}{\pgf@ya}}
\pgfpathlineto{\pgfpoint{\pgf@xa - 5pt + \@tempdima}{\pgf@ya}}
\pgfpathclose
}
}}

\boxshape{NEbox}{0pt}{5pt}
\boxshape{SEbox}{0pt}{-5pt}
\boxshape{NWbox}{5pt}{0pt}
\boxshape{SWbox}{-5pt}{0pt}
\boxshape{EBox}{-3pt}{3pt}
\boxshape{WBox}{3pt}{-3pt}
\makeatother

\tikzstyle{cloud}=[shape=cloud,draw,minimum width=1.5cm,minimum height=1.5cm]

\tikzstyle{map}=[draw,shape=NEbox,inner sep=2pt,minimum height=6mm,fill=white]
\tikzstyle{dashedmap}=[draw,dashed,shape=NEbox,inner sep=2pt,minimum height=6mm,fill=white]
\tikzstyle{mapdag}=[draw,shape=SEbox,inner sep=2pt,minimum height=6mm,fill=white]
\tikzstyle{mapadj}=[draw,shape=SEbox,inner sep=2pt,minimum height=6mm,fill=white]
\tikzstyle{maptrans}=[draw,shape=SWbox,inner sep=2pt,minimum height=6mm,fill=white]
\tikzstyle{mapconj}=[draw,shape=NWbox,inner sep=2pt,minimum height=6mm,fill=white]

\tikzstyle{medium map}=[draw,shape=NEbox,inner sep=2pt,minimum height=6mm,fill=white,minimum width=7mm]
\tikzstyle{medium map dag}=[draw,shape=SEbox,inner sep=2pt,minimum height=6mm,fill=white,minimum width=7mm]
\tikzstyle{medium map adj}=[draw,shape=SEbox,inner sep=2pt,minimum height=6mm,fill=white,minimum width=7mm]
\tikzstyle{medium map trans}=[draw,shape=SWbox,inner sep=2pt,minimum height=6mm,fill=white,minimum width=7mm]
\tikzstyle{medium map conj}=[draw,shape=NWbox,inner sep=2pt,minimum height=6mm,fill=white,minimum width=7mm]
\tikzstyle{semilarge map}=[draw,shape=NEbox,inner sep=2pt,minimum height=6mm,fill=white,minimum width=9.5mm]
\tikzstyle{semilarge map trans}=[draw,shape=SWbox,inner sep=2pt,minimum height=6mm,fill=white,minimum width=9.5mm]
\tikzstyle{semilarge map adj}=[draw,shape=SEbox,inner sep=2pt,minimum height=6mm,fill=white,minimum width=9.5mm]
\tikzstyle{semilarge map dag}=[draw,shape=SEbox,inner sep=2pt,minimum height=6mm,fill=white,minimum width=9.5mm]
\tikzstyle{semilarge map conj}=[draw,shape=NWbox,inner sep=2pt,minimum height=6mm,fill=white,minimum width=9.5mm]
\tikzstyle{large map}=[draw,shape=NEbox,inner sep=2pt,minimum height=6mm,fill=white,minimum width=12mm]
\tikzstyle{large map conj}=[draw,shape=NWbox,inner sep=2pt,minimum height=6mm,fill=white,minimum width=12mm]
\tikzstyle{very large map}=[draw,shape=NEbox,inner sep=2pt,minimum height=6mm,fill=white,minimum width=17mm]

\tikzstyle{medium dmap}=[draw,doubled,shape=NEbox,inner sep=2pt,minimum height=6mm,fill=white,minimum width=7mm]
\tikzstyle{medium dmap dag}=[draw,doubled,shape=SEbox,inner sep=2pt,minimum height=6mm,fill=white,minimum width=7mm]
\tikzstyle{medium dmap adj}=[draw,doubled,shape=SEbox,inner sep=2pt,minimum height=6mm,fill=white,minimum width=7mm]
\tikzstyle{medium dmap trans}=[draw,doubled,shape=SWbox,inner sep=2pt,minimum height=6mm,fill=white,minimum width=7mm]
\tikzstyle{medium dmap conj}=[draw,doubled,shape=NWbox,inner sep=2pt,minimum height=6mm,fill=white,minimum width=7mm]
\tikzstyle{semilarge dmap}=[draw,doubled,shape=NEbox,inner sep=2pt,minimum height=6mm,fill=white,minimum width=9.5mm]
\tikzstyle{semilarge dmap trans}=[draw,doubled,shape=SWbox,inner sep=2pt,minimum height=6mm,fill=white,minimum width=9.5mm]
\tikzstyle{semilarge dmap adj}=[draw,doubled,shape=SEbox,inner sep=2pt,minimum height=6mm,fill=white,minimum width=9.5mm]
\tikzstyle{semilarge dmap dag}=[draw,doubled,shape=SEbox,inner sep=2pt,minimum height=6mm,fill=white,minimum width=9.5mm]
\tikzstyle{semilarge dmap conj}=[draw,doubled,shape=NWbox,inner sep=2pt,minimum height=6mm,fill=white,minimum width=9.5mm]
\tikzstyle{large dmap}=[draw,doubled,shape=NEbox,inner sep=2pt,minimum height=6mm,fill=white,minimum width=12mm]
\tikzstyle{large dmap conj}=[draw,doubled,shape=NWbox,inner sep=2pt,minimum height=6mm,fill=white,minimum width=12mm]
\tikzstyle{large dmap trans}=[draw,doubled,shape=SWbox,inner sep=2pt,minimum height=6mm,fill=white,minimum width=12mm]
\tikzstyle{large dmap adj}=[draw,doubled,shape=SEbox,inner sep=2pt,minimum height=6mm,fill=white,minimum width=12mm]
\tikzstyle{large dmap dag}=[draw,doubled,shape=SEbox,inner sep=2pt,minimum height=6mm,fill=white,minimum width=12mm]
\tikzstyle{very large dmap}=[draw,doubled,shape=NEbox,inner sep=2pt,minimum height=6mm,fill=white,minimum width=19.5mm]

\tikzstyle{muxbox}=[draw,shape=rectangle,minimum height=3mm,minimum width=3mm,fill=white]
\tikzstyle{dmuxbox}=[muxbox,doubled]

\tikzstyle{box}=[draw,shape=rectangle,inner sep=2pt,minimum height=6mm,minimum width=6mm,fill=white]
\tikzstyle{dbox}=[draw,doubled,shape=rectangle,inner sep=2pt,minimum height=6mm,minimum width=6mm,fill=white]
\tikzstyle{dmap}=[draw,doubled,shape=NEbox,inner sep=2pt,minimum height=6mm,fill=white]
\tikzstyle{dmapdag}=[draw,doubled,shape=SEbox,inner sep=2pt,minimum height=6mm,fill=white]
\tikzstyle{dmapadj}=[draw,doubled,shape=SEbox,inner sep=2pt,minimum height=6mm,fill=white]
\tikzstyle{dmaptrans}=[draw,doubled,shape=SWbox,inner sep=2pt,minimum height=6mm,fill=white]
\tikzstyle{dmapconj}=[draw,doubled,shape=NWbox,inner sep=2pt,minimum height=6mm,fill=white]

\tikzstyle{ddmap}=[draw,doubled,dashed,shape=NEbox,inner sep=2pt,minimum height=6mm,fill=white]
\tikzstyle{ddmapdag}=[draw,doubled,dashed,shape=SEbox,inner sep=2pt,minimum height=6mm,fill=white]
\tikzstyle{ddmapadj}=[draw,doubled,dashed,shape=SEbox,inner sep=2pt,minimum height=6mm,fill=white]
\tikzstyle{ddmaptrans}=[draw,doubled,dashed,shape=SWbox,inner sep=2pt,minimum height=6mm,fill=white]
\tikzstyle{ddmapconj}=[draw,doubled,dashed,shape=NWbox,inner sep=2pt,minimum height=6mm,fill=white]

\boxshape{sNEbox}{0pt}{3pt}
\boxshape{sSEbox}{0pt}{-3pt}
\boxshape{sNWbox}{3pt}{0pt}
\boxshape{sSWbox}{-3pt}{0pt}
\tikzstyle{smap}=[draw,shape=sNEbox,fill=white]
\tikzstyle{smapdag}=[draw,shape=sSEbox,fill=white]
\tikzstyle{smapadj}=[draw,shape=sSEbox,fill=white]
\tikzstyle{smaptrans}=[draw,shape=sSWbox,fill=white]
\tikzstyle{smapconj}=[draw,shape=sNWbox,fill=white]

\tikzstyle{dsmap}=[draw,dashed,shape=sNEbox,fill=white]
\tikzstyle{dsmapdag}=[draw,dashed,shape=sSEbox,fill=white]
\tikzstyle{dsmaptrans}=[draw,dashed,shape=sSWbox,fill=white]
\tikzstyle{dsmapconj}=[draw,dashed,shape=sNWbox,fill=white]

\boxshape{mNEbox}{0pt}{10pt}
\boxshape{mSEbox}{0pt}{-10pt}
\boxshape{mNWbox}{10pt}{0pt}
\boxshape{mSWbox}{-10pt}{0pt}
\tikzstyle{mmap}=[draw,shape=mNEbox]
\tikzstyle{mmapdag}=[draw,shape=mSEbox]
\tikzstyle{mmaptrans}=[draw,shape=mSWbox]
\tikzstyle{mmapconj}=[draw,shape=mNWbox]

\tikzstyle{mmapgray}=[draw,fill=gray!40!white,shape=mNEbox]
\tikzstyle{smapgray}=[draw,fill=gray!40!white,shape=sNEbox]

\makeatletter

\pgfdeclareshape{cornerpoint}{
\inheritsavedanchors[from=rectangle] 
\inheritanchorborder[from=rectangle]
\inheritanchor[from=rectangle]{center}
\inheritanchor[from=rectangle]{north}
\inheritanchor[from=rectangle]{south}
\inheritanchor[from=rectangle]{west}
\inheritanchor[from=rectangle]{east}
\backgroundpath{
\southwest \pgf@xa=\pgf@x \pgf@ya=\pgf@y
\northeast \pgf@xb=\pgf@x \pgf@yb=\pgf@y

\pgfmathsetmacro{\pgf@shorten@left}{\pgfkeysvalueof{/tikz/shorten left}}
\pgfmathsetmacro{\pgf@shorten@right}{\pgfkeysvalueof{/tikz/shorten right}}

\pgfpathmoveto{\pgfpoint{0.5 * (\pgf@xa + \pgf@xb)}{\pgf@ya - 5pt}}
\pgfpathlineto{\pgfpoint{\pgf@xa - 8pt + \pgf@shorten@left}{\pgf@yb - 1.5 * \pgf@shorten@left}}
\pgfpathlineto{\pgfpoint{\pgf@xa - 8pt + \pgf@shorten@left}{\pgf@yb}}
\pgfpathlineto{\pgfpoint{\pgf@xb + 8pt - \pgf@shorten@right}{\pgf@yb}}
\pgfpathlineto{\pgfpoint{\pgf@xb + 8pt - \pgf@shorten@right}{\pgf@yb - 1.5 * \pgf@shorten@right}}
\pgfpathclose
}
}

\pgfdeclareshape{cornercopoint}{
\inheritsavedanchors[from=rectangle] 
\inheritanchorborder[from=rectangle]
\inheritanchor[from=rectangle]{center}
\inheritanchor[from=rectangle]{north}
\inheritanchor[from=rectangle]{south}
\inheritanchor[from=rectangle]{west}
\inheritanchor[from=rectangle]{east}
\backgroundpath{
\southwest \pgf@xa=\pgf@x \pgf@ya=\pgf@y
\northeast \pgf@xb=\pgf@x \pgf@yb=\pgf@y

\pgfmathsetmacro{\pgf@shorten@left}{\pgfkeysvalueof{/tikz/shorten left}}
\pgfmathsetmacro{\pgf@shorten@right}{\pgfkeysvalueof{/tikz/shorten right}}

\pgfpathmoveto{\pgfpoint{0.5 * (\pgf@xa + \pgf@xb)}{\pgf@yb + 5pt}}
\pgfpathlineto{\pgfpoint{\pgf@xa - 8pt + \pgf@shorten@left}{\pgf@ya + 1.5 * \pgf@shorten@left}}
\pgfpathlineto{\pgfpoint{\pgf@xa - 8pt + \pgf@shorten@left}{\pgf@ya}}
\pgfpathlineto{\pgfpoint{\pgf@xb + 8pt - \pgf@shorten@right}{\pgf@ya}}
\pgfpathlineto{\pgfpoint{\pgf@xb + 8pt - \pgf@shorten@right}{\pgf@ya + 1.5 * \pgf@shorten@right}}
\pgfpathclose
}
}

\makeatother

\pgfkeyssetvalue{/tikz/shorten left}{0pt}
\pgfkeyssetvalue{/tikz/shorten right}{0pt}

\tikzstyle{kpoint common}=[draw,fill=white,inner sep=1pt,minimum height=4mm]
\tikzstyle{kpoint sc}=[shape=cornerpoint,kpoint common]
\tikzstyle{kpoint adjoint sc}=[shape=cornercopoint,kpoint common]
\tikzstyle{kpoint}=[shape=cornerpoint,shorten left=5pt,kpoint common]
\tikzstyle{kpoint adjoint}=[shape=cornercopoint,shorten left=5pt,kpoint common]
\tikzstyle{kpoint conjugate}=[shape=cornerpoint,shorten right=5pt,kpoint common]
\tikzstyle{kpoint transpose}=[shape=cornercopoint,shorten right=5pt,kpoint common]
\tikzstyle{kpoint symm}=[shape=cornerpoint,shorten left=5pt,shorten right=5pt,kpoint common]

\tikzstyle{wide kpoint sc}=[shape=cornerpoint,kpoint common, minimum width=1 cm]
\tikzstyle{wide kpointdag sc}=[shape=cornercopoint,kpoint common, minimum width=1 cm]

\tikzstyle{black kpoint}=[shape=cornerpoint,shorten left=5pt,kpoint common,fill=black,font=\color{white}]

\tikzstyle{black kpoint sm}=[shape=cornerpoint,shorten left=5pt,kpoint common,fill=black,font=\color{white},scale=0.75]

\tikzstyle{black kpoint adjoint}=[shape=cornercopoint,shorten left=5pt,kpoint common,fill=black,font=\color{white}]
\tikzstyle{black kpointadj}=[shape=cornercopoint,shorten left=5pt,kpoint common,fill=black,font=\color{white}]

\tikzstyle{black kpointadj sm}=[shape=cornercopoint,shorten left=5pt,kpoint common,fill=black,font=\color{white},scale=0.75]

\tikzstyle{black dkpoint}=[shape=cornerpoint,shorten left=5pt,kpoint common,fill=black, doubled,font=\color{white}]
\tikzstyle{black dkpoint adjoint}=[shape=cornercopoint,shorten left=5pt,kpoint common,fill=black, doubled,font=\color{white}]
\tikzstyle{black dkpointadj}=[shape=cornercopoint,shorten left=5pt,kpoint common,fill=black, doubled,font=\color{white}]

\tikzstyle{black dkpoint sm}=[shape=cornerpoint,shorten left=5pt,kpoint common,fill=black, doubled,font=\color{white},scale=0.75]
\tikzstyle{black dkpointadj sm}=[shape=cornercopoint,shorten left=5pt,kpoint common,fill=black, doubled,font=\color{white},scale=0.75]

\tikzstyle{kpointdag}=[kpoint adjoint]
\tikzstyle{kpointadj}=[kpoint adjoint]
\tikzstyle{kpointconj}=[kpoint conjugate]
\tikzstyle{kpointtrans}=[kpoint transpose]

\tikzstyle{big kpoint}=[kpoint, minimum width=1.2 cm, minimum height=8mm, inner sep=4pt, text depth=3mm]

\tikzstyle{wide kpoint}=[kpoint, minimum width=1 cm, inner sep=2pt]
\tikzstyle{wide kpointdag}=[kpointdag, minimum width=1 cm, inner sep=2pt]
\tikzstyle{wide kpointconj}=[kpointconj, minimum width=1 cm, inner sep=2pt]
\tikzstyle{wide kpointtrans}=[kpointtrans, minimum width=1 cm, inner sep=2pt]

\tikzstyle{wider kpoint}=[kpoint, minimum width=1.25 cm, inner sep=2pt]
\tikzstyle{wider kpointdag}=[kpointdag, minimum width=1.25 cm, inner sep=2pt]
\tikzstyle{wider kpointconj}=[kpointconj, minimum width=1.25 cm, inner sep=2pt]
\tikzstyle{wider kpointtrans}=[kpointtrans, minimum width=1.25 cm, inner sep=2pt]

\tikzstyle{gray kpoint}=[kpoint,fill=gray!50!white]
\tikzstyle{gray kpointdag}=[kpointdag,fill=gray!50!white]
\tikzstyle{gray kpointadj}=[kpointadj,fill=gray!50!white]
\tikzstyle{gray kpointconj}=[kpointconj,fill=gray!50!white]
\tikzstyle{gray kpointtrans}=[kpointtrans,fill=gray!50!white]

\tikzstyle{gray dkpoint}=[kpoint,fill=gray!50!white,doubled]
\tikzstyle{gray dkpointdag}=[kpointdag,fill=gray!50!white,doubled]
\tikzstyle{gray dkpointadj}=[kpointadj,fill=gray!50!white,doubled]
\tikzstyle{gray dkpointconj}=[kpointconj,fill=gray!50!white,doubled]
\tikzstyle{gray dkpointtrans}=[kpointtrans,fill=gray!50!white,doubled]

\tikzstyle{white label}=[draw,fill=white,rectangle,inner sep=0.7 mm]
\tikzstyle{gray label}=[draw,fill=gray!50!white,rectangle,inner sep=0.7 mm]
\tikzstyle{black label}=[draw,fill=black,rectangle,inner sep=0.7 mm]

\tikzstyle{dkpoint}=[kpoint,doubled]
\tikzstyle{wide dkpoint}=[wide kpoint,doubled]
\tikzstyle{dkpointdag}=[kpoint adjoint,doubled]
\tikzstyle{wide dkpointdag}=[wide kpointdag,doubled]
\tikzstyle{dkcopoint}=[kpoint adjoint,doubled]
\tikzstyle{dkpointadj}=[kpoint adjoint,doubled]
\tikzstyle{dkpointconj}=[kpoint conjugate,doubled]
\tikzstyle{dkpointtrans}=[kpoint transpose,doubled]

\tikzstyle{kscalar}=[kpoint common, shape=EBox, inner xsep=-1pt, inner ysep=3pt,font=\small]
\tikzstyle{kscalarconj}=[kpoint common, shape=WBox, inner xsep=-1pt, inner ysep=3pt,font=\small]

\tikzstyle{spekpoint}=[kpoint sc,minimum height=5mm,inner sep=3pt]
\tikzstyle{spekcopoint}=[kpoint adjoint sc,minimum height=5mm,inner sep=3pt]

\tikzstyle{dspekpoint}=[spekpoint,doubled]
\tikzstyle{dspekcopoint}=[spekcopoint,doubled]

 \tikzstyle{upground}=[circuit ee IEC,thick,ground,rotate=90,scale=2.5]
 \tikzstyle{downground}=[circuit ee IEC,thick,ground,rotate=-90,scale=2.5]
 \tikzstyle{bigground}=[regular polygon,regular polygon sides=3,draw=gray,scale=0.50,inner sep=-0.5pt,minimum width=10mm,fill=gray]

\tikzstyle{arrs}=[-latex,font=\small,auto]
\tikzstyle{arrow plain}=[arrs]
\tikzstyle{arrow dashed}=[dashed,arrs]
\tikzstyle{arrow bold}=[very thick,arrs]
\tikzstyle{arrow hide}=[draw=white!0,-]
\tikzstyle{arrow reverse}=[latex-]
\tikzstyle{cdnode}=[]

\let\olddagger\dagger
\renewcommand{\dagger}{\ensuremath{\olddagger}\xspace}

\usepackage{makeidx}
\makeindex

\newtheorem{thm}{Theorem}[section]
\newtheorem{proposition}[thm]{Proposition}
\newtheorem{definition}[thm]{Definition}
\newtheorem{example}[thm]{Example}

\newtheorem{example*}[thm]{Example*}
\newtheorem{examples*}[thm]{Examples*}
\newtheorem{remark}[thm]{Remark}
\newtheorem{remark*}[thm]{Remark*}

\usepackage{color}
\def\bR{\begin{color}{red}}
\def\bB{\begin{color}{blue}}
\def\bM{\begin{color}{magenta}}
\def\bC{\begin{color}{cyan}}
\def\bW{\begin{color}{white}}
\def\bBl{\begin{color}{black}}
\def\bG{\begin{color}{green}}
\def\bY{\begin{color}{yellow}}
\def\e{\end{color}\xspace}
\newcommand{\bit}{\begin{itemize}}
\newcommand{\eit}{\end{itemize}\par\noindent}
\newcommand{\ben}{\begin{enumerate}}
\newcommand{\een}{\end{enumerate}\par\noindent}
\newcommand{\beq}{\begin{equation}}
\newcommand{\eeq}{\end{equation}\par\noindent}
\newcommand{\beqa}{\begin{eqnarray*}}
\newcommand{\eeqa}{\end{eqnarray*}\par\noindent}
\newcommand{\beqn}{\begin{eqnarray}}
\newcommand{\eeqn}{\end{eqnarray}\par\noindent}

\def\jR{\begin{color}{purple}}
\def\jB{\begin{color}{teal}}
\def\jM{\begin{color}{magenta}}
\def\jC{\begin{color}{cyan}}
\def\jW{\begin{color}{white}}
\def\jBl{\begin{color}{black}}
\def\jG{\begin{color}{green}}
\def\jY{\begin{color}{yellow}}

\newcommand{\quality}{
\mathcal{Q}\hspace{-1mm}\left[\ %
\InputIfFileExists{Diagrams/plainleaksmall.tikz}{}{\input{./figures/Diagrams/plainleaksmall.tikz}}\ \right]
}

\begin{document}

\begingroup
\centering
{\Large\textbf{Leaks: quantum, classical, intermediate, and more} \\[1.5em]
 \normalsize  John H. Selby\textsuperscript{$\ast,\dagger$,}\footnote{Electronic address: john.selby08@imperial.ac.uk} and Bob Coecke\textsuperscript{$\dagger$}
}
\\[1em]
 \it \textsuperscript{$\ast$} Department of Physics, Imperial College London,  London SW7 2AZ, UK. \\  \it \textsuperscript{$\dagger$} Department of Computer Science, University of Oxford, OX1 3QD, UK. \\

\endgroup

\begin{abstract}
We introduce the notion of a leak for general process theories and identify quantum theory as a theory with minimal leakage, while classical theory has maximal leakage. We provide a~construction that adjoins leaks to theories, an instance of which describes the emergence of classical theory by adjoining decoherence leaks to quantum theory. Finally, we show that defining a notion of purity for processes in general process theories has to make reference to the leaks of that theory, a feature missing in standard definitions; hence, we propose a refined definition and study the resulting notion of purity for quantum, classical and intermediate theories.
\end{abstract}

\section{Introduction}

Can we explain why the world is quantum by finding some sense in which quantum theory is an optimal theory? Broadcasting distinguishes quantum theory from classical theory in that quantum states cannot be broadcast \cite{Nobroadcast}, but neither can the states of many other theories \cite{BBLW, CKpaperI}. Non-locality is a measure of non-classicality, and quantum theory is non-local, but not maximally so \cite{PRBox}. Therefore, is there some manner in which we can uniquely single out quantum theory? In this paper, we show that quantum theory is a leak-free theory, whilst classical theory is maximally leaking. We formalise the notion of a leak, which can roughly be thought of as a `one-sided broadcasting map', within the process-theoretic framework \cite{AC1, CKpaperI, CKBook} as a particular type of process, which, as the name suggests, accounts for leaking state-data into the environment.

Moreover, there is a natural way to introduce leaks to any theory, and by doing so, we obtain new theories. We call this the leak construction. In particular, classical theory can be obtained from quantum theory in this manner, where, in this example, the leaking is then nothing but decoherence \cite{kuperberg2003capacity, zurek2009quantum}. Hence, the concept of a leak allows us to generalise decoherence to arbitrary process theories. Besides classical theory, any theory characterised by some finite-dimensional C*-algebra can be obtained in this manner from quantum theory. In fact, as we show in a follow-up paper \cite{CST}, only C*-algebras can be obtained in this manner. Leaks therefore capture the operational content of finite-dimensional C*-algebras on-the-nose, in a manner that does not involve any additive structure, nor a $*$-operation.

Finally, we observe that defining {purity of processes} in process theories with leaks is problematic; in~particular, this is the case for classical theory. Making explicit use of the concept of a leak, we therefore propose a new definition that makes sense for arbitrary processes in arbitrary process~theories.

\subsection*{Related Work} 

As explained in detail in the follow-up paper \cite{CST}, the leak construction is related to the ``constructions of classical system types'' in \cite{SelingerIdempotent,heunen2013completely,cunningham2015axiomatizing}. More specifically, in the case of quantum theory, we exactly obtain the same result, but in a much simpler way, with much less use of structure and~guided by a clear operational meaning.
The notion of a leak is closely related to the decomposability of a state-space \cite{Barrett} in the generalised probabilistic theory framework as, at least under some standard assumptions, such as the ``no-restriction hypothesis'', each is equivalent to the existence of a non-disturbing measurement as discussed in \cite{richens2016entanglement}.

\section{Process Theories with Discarding...
}

A process theory \cite{CKpaperI, CKBook} is a collection of systems that are represented by wires and processes that are represented by boxes with wires as inputs (at the bottom) and outputs (at the top). Moreover, when we plug these boxes together:
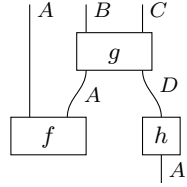
\[
\InputIfFileExists{Diagrams/compound-process.tikz}{}{\input{./figures/Diagrams/compound-process.tikz}}
\]
the resulting diagram should also be a process. To be mathematically more precise, the data that make up a diagram are:
\bit
\item the boxes that appear in the diagram and
\item how these boxes are wired together, including the overall ordering of inputs/outputs.
\eit
Hence, two diagrams are equal when these data match up.

By a circuit \cite{CKpaperI, CKBook}, we mean a diagram that can be constructed by means of the obvious operations of parallel composition $\otimes$ and sequential composition $\circ$ of boxes. For example, the following diagrams is a circuit:
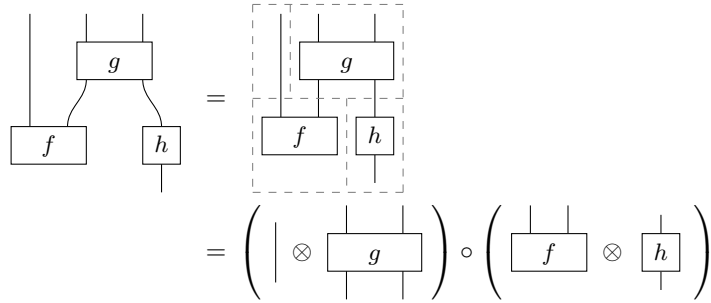
\begin{align*}
\InputIfFileExists{Diagrams/compound-process-notypes.tikz}{}{\input{./figures/Diagrams/compound-process-notypes.tikz}}\ \
& = \ \ %
\InputIfFileExists{Diagrams/funcex.tikz}{}{\input{./figures/Diagrams/funcex.tikz}} \vspace{2.5mm}\\
& = \, \left(\ \shortidwire \, \otimes\ %
\InputIfFileExists{Diagrams/compose-g-bis.tikz}{}{\input{./figures/Diagrams/compose-g-bis.tikz}}\ \right) \circ\left(\ %
\InputIfFileExists{Diagrams/compose-f-bis.tikz}{}{\input{./figures/Diagrams/compose-f-bis.tikz}} \ \otimes\ \boxmap{h}\ \right)
\end{align*}

Composite systems, denoted $A\otimes B$, then simply arise by pairing wires:
\[
\begin{tikzpicture}
	\begin{pgfonlayer}{nodelayer}
		\node [style=none] (0) at (0, -0.75) {};
		\node [style=none] (1) at (0, 0.75) {};
		\node [style=right label, yshift=-2] (2) at (0.25, 0) {$A\otimes B$};
	\end{pgfonlayer}
	\begin{pgfonlayer}{edgelayer}
		\draw [style=swap] (1) to (0.center);
	\end{pgfonlayer}
\end{tikzpicture}} \ \ := \ \ %
\begin{tikzpicture}
	\begin{pgfonlayer}{nodelayer}
		\node [style=none] (0) at (0, -0.75) {};
		\node [style=none] (1) at (0, 0.75) {};
		\node [style=right label, yshift=-2] (2) at (0.25, 0) {$A$};
	\end{pgfonlayer}
	\begin{pgfonlayer}{edgelayer}
		\draw [style=swap] (1) to (0.center);
	\end{pgfonlayer}
\end{tikzpicture}}\ %
\begin{tikzpicture}
	\begin{pgfonlayer}{nodelayer}
		\node [style=none] (0) at (0, -0.75) {};
		\node [style=none] (1) at (0, 0.75) {};
		\node [style=right label, yshift=-2] (2) at (0.25, 0) {$B$};
	\end{pgfonlayer}
	\begin{pgfonlayer}{edgelayer}
		\draw [style=swap] (1) to (0.center);
	\end{pgfonlayer}
\end{tikzpicture}}
\]

\begin{remark}
A process theory with circuits as diagrams can also be defined as a strict symmetric monoidal category. Strictness means that associativities and unit laws hold on-the-nose, unlike the symmetric monoidal categories of concrete mathematical models where non-trivial associativity and unit natural isomorphisms are required. Fortunately, by Mac Lane's strictification theorem \cite{MacLane}, every such category is categorically equivalent (although not isomorphic) to a strict one, which means that for all practical purposes, it can be thought of as a~strict~one.
\end{remark}

A state is a process without inputs; an effect is a process without outputs; and a number is a process with neither inputs nor outputs. One special number is the empty diagram:
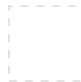
\[
\emptydiag
\]
which in most theories coincides with the number one.

Throughout this paper, for each system in a process theory, we postulate the existence of a~discarding effect, which is interpreted just as its name indicates and which is denoted as:
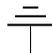
\[
\begin{tikzpicture}
	\begin{pgfonlayer}{nodelayer}
		\node [style=none] (0) at (0, -0.75) {};
		\node [style=upground] (1) at (0, 0.25) {};
	\end{pgfonlayer}
	\begin{pgfonlayer}{edgelayer}
		\draw [style=swap] (1) to (0.center);    
	\end{pgfonlayer} 
\end{tikzpicture}}
\]
We also make the natural assumption that discarding effects compose:
\beq\label{eq:disccomp}
\begin{tikzpicture}
	\begin{pgfonlayer}{nodelayer}
		\node [style=none] (0) at (0, -0.75) {};
		\node [style=upground] (1) at (0, 0.25) {};
		\node [style=right label] (2) at (0.25, -0.5) {$A\otimes B$};
	\end{pgfonlayer}
	\begin{pgfonlayer}{edgelayer}
		\draw [style=swap] (1) to (0.center);
	\end{pgfonlayer}
\end{tikzpicture}} \ \ := \ \ %
\begin{tikzpicture}
	\begin{pgfonlayer}{nodelayer}
		\node [style=none] (0) at (0, -0.75) {};
		\node [style=upground] (1) at (0, 0.25) {};
		\node [style=right label] (2) at (0.25, -0.5) {$A$};
	\end{pgfonlayer}
	\begin{pgfonlayer}{edgelayer}
		\draw [style=swap] (1) to (0.center);
	\end{pgfonlayer}
\end{tikzpicture}}\ %
\begin{tikzpicture}
	\begin{pgfonlayer}{nodelayer}
		\node [style=none] (0) at (0, -0.75) {};
		\node [style=upground] (1) at (0, 0.25) {};
		\node [style=right label] (2) at (0.25, -0.5) {$B$};
	\end{pgfonlayer}
	\begin{pgfonlayer}{edgelayer}
		\draw [style=swap] (1) to (0.center);
	\end{pgfonlayer}
\end{tikzpicture}}
\eeq
A process $f$ is causal if we have:
\beq\label{eq:term}
\InputIfFileExists{Diagrams/norm.tikz}{}{\input{./figures/Diagrams/norm.tikz}}
\eeq
and a theory is causal if all of the processes of the theory are causal. Therefore, except for the fact that it composes, discarding is not subject to any defining constraints. In a sense, its behaviour is entirely implicit within its role within the defining equation of causality. In particular, by Equation (\ref{eq:term}) where $f$ is taken to be an effect, it immediately follows that the only effects in a causal theory are the discarding effects. In this form, the axiom of causality traces back to \cite{Chiri1}. When restricting to causal processes, a process theory is non-signalling \cite{Cnonsig}; hence, the causality of a theory is vital to guarantee compatibility with relativity.

\begin{example}[Classical probability theory]\label{ex:classtheory}
When viewing probability theory as a process theory, systems are $n$-state classical systems and boxes are $n\times m$ stochastic matrices, and so, in particular, states are probability distributions. Discarding is given by marginalisation, and so, causality boils down to the fact that the entries of a~probability distribution add up to one and that the entries in each column of a stochastic matrix add up to one.
\end{example}

\begin{example}[Quantum theory]\label{ex:quantumtheory}
Quantum theory as a process theory has finite dimensional Hilbert spaces ${\cal H}$ as its systems and completely positive trace preserving (CPTP) maps:
\[
\xi: {\cal B(H)}\to {\cal B(H')}
\]
as its processes. Causality for density operators means having trace one and for completely positive maps means being trace-preserving. One can also include classical data as additional systems, and then measurements and controlled operations are also processes. If this is the case, we will often denote the classical systems as dotted wires to distinguish them from quantum wires. Specifically, measurements are processes from quantum to classical systems where the probabilities of obtaining the different outcomes are encoded in the classical system. Causality then implies that, for projective measurements, the projectors form a resolution of the identity and, for general measurements that the POVM elements sum to discarding. A full description and a pedagogical introduction to this theory is in \cite{CKpaperI, CQMII, CKBook}.
\end{example}

Typically, as will be the case in the examples below, we will want to describe both causal and non-causal processes. We therefore will still, for each system, have a discarding map, which specifies the causal processes, but there will also be other processes that will not satisfy Equation (\ref{eq:term}). There are two main reasons for this. The first is to allow us to discuss {events}, i.e.,~processes that we cannot make happen deterministically, but that can occur as a particular outcome in some experiment; therefore, allowing us to obtain the probability of obtaining a specific outcome, which, in particular, allows us, via suitable renormalisation, to describe {post-selection}.
The second reason is mathematical simplicity: it is often much easier to define the process theory, or various structures within it, in the non-causal setting and then to restrict to the causal sub-theory when necessary.

\begin{example}[Non-causal extension of quantum theory]\label{Ex:NonCausalQ}
To describe non-causal processes in quantum theory, rather than taking processes as completely positive trace-preserving maps, we instead just require that they are completely positive. It is very standard within quantum theory to consider such processes, for example Dirac bras are non-causal or, more generally, individual POVM elements are non-causal.

An important tool is the Choi--Jamiolkowski isomorphism between transformations and bipartite states. One direction of this isomorphism can be realised causally, using the Bell state, which we represent, up to a~normalisation factor, with a cup-shaped wire:
\[
{1\over D}\ \, %
\begin{tikzpicture}
	\begin{pgfonlayer}{nodelayer}
		\node [style=none] (0) at (-0.75, 0.5) {};
		\node [style=none] (1) at (0.75, 0.5) {};
		\node [style=none] (2) at (0.75, 0.5) {};
		\node [style=none] (3) at (-0.75, 0.5) {}; 
	\end{pgfonlayer}
	\begin{pgfonlayer}{edgelayer}
		\draw [in=-90, out=-90, looseness=1.75] (0.center) to (1.center);
	\end{pgfonlayer}
\end{tikzpicture}} \qquad\qquad\qquad \text{where} \qquad\qquad\qquad D:=%
\begin{tikzpicture}
	\begin{pgfonlayer}{nodelayer}
		\node [style=none] (0) at (1.5, 0.25) {};
		\node [style=none] (1) at (0, 0.25) {};
		\node [style=upground] (2) at (0, 0.5000001) {};
		\node [style=upground] (3) at (1.5, 0.5000001) {};
	\end{pgfonlayer}
	\begin{pgfonlayer}{edgelayer}
		\draw [ bend right=90, looseness=1.50] (1.center) to (0.center);
	\end{pgfonlayer}
\end{tikzpicture}
}
\]
which allows us to ``bend wires up'':
\[
\InputIfFileExists{Diagrams/purefadj.tikz}{}{\input{./figures/Diagrams/purefadj.tikz}}\ \ \mapsto\ \ {1\over D}\ \, %
\InputIfFileExists{Diagrams/cupType_up.tikz}{}{\input{./figures/Diagrams/cupType_up.tikz}}
\]
This associates with each (causal) process a (causal) bipartite state. The other direction is however not realisable causally, as it relies on the Bell effect, which we represent with a cap-shaped wire:
\[
\InputIfFileExists{Diagrams/bipartiteState.tikz}{}{\input{./figures/Diagrams/bipartiteState.tikz}}\ \ \mapsto\ \ D\ \, %
\InputIfFileExists{Diagrams/purefbend.tikz}{}{\input{./figures/Diagrams/purefbend.tikz}}
\]
The fact that this is an isomorphism provides us with the following intuitive diagrammatic rule (justifying the representation of these as a cup and cap):
\beq \label{eq:yank}
\InputIfFileExists{Diagrams/wire_yank_all.tikz}{}{\input{./figures/Diagrams/wire_yank_all.tikz}}
\eeq
It is then clear that the cap cannot be causal (even up to a rescaling) as, if it were, then the identity transformation would be separable, i.e.:
\[
\InputIfFileExists{Diagrams/capNotDiscard.tikz}{}{\input{./figures/Diagrams/capNotDiscard.tikz}}
\]
where in the second step we relied on the fact that by causality, all effects must be discarding, so in particular, the cap, as well as Equation (\ref{eq:disccomp}).
\end{example}

\begin{example}[Non-causal extension of classical theory]\label{Ex:NonCausalCl}
We can similarly extend classical theory, taking processes as $n\times m$ matrices with positive real elements as opposed to stochastic matrices. This again allows us to discuss particular outcomes of measurements, which may not happen with certainty, and moreover, gives us a~classical equivalent of the Choi--Jamiolkowski isomorphism where rather than using the Bell state and effect, we use the perfectly correlated state and effect, again denoted by a cup and a cap. These can be defined in terms of the orthonormal basis states and effects as:
\[{1\over n} \ %
\begin{tikzpicture}
	\begin{pgfonlayer}{nodelayer}
		\node [style=copoint] (0) at (0, -0) {$i$};
		\node [style=copoint] (1) at (1.5, -0) {$j$};
	\end{pgfonlayer}
	\begin{pgfonlayer}{edgelayer}
		\draw [bend right=90, looseness=1.75] (0) to (1);
	\end{pgfonlayer}
\end{tikzpicture}} \ \ = \ \ {\delta_{ij} \over n}\qquad \qquad\qquad \text{and} \qquad\qquad \qquad %
\begin{tikzpicture}
	\begin{pgfonlayer}{nodelayer}
		\node [style=point] (0) at (0, -0) {$i$};
		\node [style=point] (1) at (1.5, -0) {$j$};
	\end{pgfonlayer}
	\begin{pgfonlayer}{edgelayer}
		\draw [bend left=90, looseness=1.75] (0) to (1);
	\end{pgfonlayer}
\end{tikzpicture}}\ \ = \ \ \delta_{ij}\]
respectively. It is simple to check that these also satisfy Equation (\ref{eq:yank}) as we would expect from the choice of the diagrammatic representation.
\end{example}

This forms the basic structures needed to describe the physical content of process theories; however, we will need some further tools for the proofs. These are all defined in the standard way for categorical quantum mechanics
and surveyed in Appendix \ref{mathsAppendix} for those unfamiliar with the field.

\section{... 
and Leaks}
\begin{definition}
A leak is a process:
\beq\label{eq:leakL}
\InputIfFileExists{Diagrams/leakL.tikz}{}{\input{./figures/Diagrams/leakL.tikz}}
\eeq
which has discarding as a right counit, that is:
\beq\label{eq:leakdisc1}
\InputIfFileExists{Diagrams/leakdisc1.tikz}{}{\input{./figures/Diagrams/leakdisc1.tikz}}
\eeq
\end{definition}

\begin{proposition}
All leaks are causal.
\end{proposition}
\begin{proof}
Causality of a leak means:
\ctikzfig{Diagrams/leakdisc1proof}
and this equation is obtained by discarding the outputs in (\ref{eq:leakdisc1}).
\end{proof}

When we have multiple leaks around, we may often represent them with different colours to distinguish them.

\begin{proposition}
Leaks compose to give leaks.
\end{proposition}
\proof
Sequential composition of leaks is again a leak:
\[%
\InputIfFileExists{Diagrams/leakComposition1.tikz}{}{\input{./figures/Diagrams/leakComposition1.tikz}}
\]
since we have:
\[
\InputIfFileExists{Diagrams/leaksCompose2.tikz}{}{\input{./figures/Diagrams/leaksCompose2.tikz}}
\]
and the same goes for parallel composition:
\[%
\InputIfFileExists{Diagrams/leaksCompose3.tikz}{}{\input{./figures/Diagrams/leaksCompose3.tikz}}\]
since we have:
\[%
\InputIfFileExists{Diagrams/leaksCompose4.tikz}{}{\input{./figures/Diagrams/leaksCompose4.tikz}}\]
\endproof

For classical probability theory, copying of support elements provides a leak:
\[
\InputIfFileExists{Diagrams/broadcast.tikz}{}{\input{./figures/Diagrams/broadcast.tikz}}\ \ : X\to X \times X:: x\mapsto (x, x)
\]
since if we discard a copy, we are back with what we started off with. In fact, strictly speaking, what we are dealing with here is not a copying operation since while it copies pure classical states, it does not do that for impure ones. What it is instead is broadcasting, that is besides Equation (\ref{eq:leakdisc1}), discarding is also a left counit for the leaking process:
\beq\label{eq:leakdisc2}
\InputIfFileExists{Diagrams/leakdisc2.tikz}{}{\input{./figures/Diagrams/leakdisc2.tikz}}
\eeq
 Note that this requires $L:=A$ in Equation (\ref{eq:leakL}). This is the {maximal} possible leak for any system, as all of the information about the ingoing state is leaked out.

On the other hand, quantum theory does not allow for broadcasting \cite{Nobroadcast}. In fact, the only kind of leak quantum theory admits is constant leaking. This immediately follows from the following fact about quantum processes, which states that any dilation of a~pure process, i.e.,~representation as a~process with an extra output that is discarded, must separate:

\begin{proposition}\label{prop:reducedprocpure}
For pure quantum processes $f$, we have:
\beq\label{eq:PurityOfProc}
\begin{tikzpicture}
	\begin{pgfonlayer}{nodelayer}
		\node [style=box] (0) at (0, 0) {$f$};
		\node [style=none] (1) at (0, 1.25) {};
		\node [style=none] (2) at (0, -1.25) {};
	\end{pgfonlayer}
	\begin{pgfonlayer}{edgelayer}
		\draw (1.center) to (0); 
		\draw (0) to (2.center); 
	\end{pgfonlayer}
\end{tikzpicture}}\ \ =\ \ %
\InputIfFileExists{Diagrams/pure1.tikz}{}{\input{./figures/Diagrams/pure1.tikz}}\qquad \Longrightarrow\qquad %
\InputIfFileExists{Diagrams/pure2.tikz}{}{\input{./figures/Diagrams/pure2.tikz}}\ \ = \ \ %
\InputIfFileExists{Diagrams/pure4.tikz}{}{\input{./figures/Diagrams/pure4.tikz}}
\eeq
with $\rho$ causal. That is, if a reduced process $f$ is pure, then the process $g$ we started from must separate.
\end{proposition}
\begin{proof}
See, e.g.,~\cite{CKpaperI, CKBook}.
\end{proof}

Hence, since the identity is pure, by Proposition \ref{prop:reducedprocpure}, it follows from the defining equation of a leak (\ref{eq:leakdisc1}) that any leak for quantum theory must be constant, that is of the form:
\beq\label{eq:cteleak1}
\begin{tikzpicture}
	\begin{pgfonlayer}{nodelayer}
		\node [style=none] (0) at (-0.5, -1) {};
		\node [style=none] (1) at (-0.5, 1) {};
		\node [style=point] (2) at (0.5, 0) {$\rho$};
		\node [style=none] (3) at (0.5, 1) {};
	\end{pgfonlayer}
	\begin{pgfonlayer}{edgelayer}
		\draw [style=swap] (1.center) to (0.center);      
		\draw (3.center) to (2);
	\end{pgfonlayer}
\end{tikzpicture}}
\eeq
where we need to take the state to be causal:
\beq\label{eq:cteleak2}
\InputIfFileExists{Diagrams/cteleak2.tikz}{}{\input{./figures/Diagrams/cteleak2.tikz}}
\eeq

\begin{remark}\label{Rem:QuantumPure}
 In quantum theory, Proposition~\ref{prop:reducedprocpure} can actually be taken as a definition of the purity of processes, that is a quantum process $f$ is pure if and only if all dilations of $f$ separate. However, in theories with non-constant leaks, this definition must be revised as we discuss in detail in Section~\ref{Sec:ProcessPurity}.
 \end{remark}

Of course, (\ref{eq:cteleak1}) is also a leak for classical probability theory, and another example arises by combining broadcasting and a constant:
\beq\label{eq:broadcastANDcte}
\InputIfFileExists{Diagrams/broadcastANDcte.tikz}{}{\input{./figures/Diagrams/broadcastANDcte.tikz}}
\eeq

 At least qualitatively, quantum theory can therefore be described as a {minimally-leaking} theory, as all leaks are constant leaks, whilst classical theory is {maximally leaking}, as for each system, there is a~maximal leak. We will now provide qualitative substance to this claim.

\section{Quality of a Leak}

 For the sake of simplicity of the argument, we will restrict ourselves to a special kind of process theories that admit the notion of a feedback wire. Explicitly spelling out the process-theoretic characterisation of a feedback wire as in \cite{JSV} goes beyond the scope of this paper. It suffices to know that they exist in both quantum and classical theory, where they can be constructed in the obvious way using the cups and caps of Examples \ref{Ex:NonCausalQ} and \ref{Ex:NonCausalCl}. The behaviour of such a feedback wire is that of a wire of the shape:
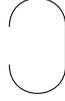
\[
\InputIfFileExists{Diagrams/trace.tikz}{}{\input{./figures/Diagrams/trace.tikz}}
\]
for which we have the obvious equations, such as:
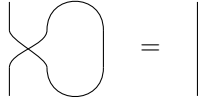
\[
\InputIfFileExists{Diagrams/traceeq.tikz}{}{\input{./figures/Diagrams/traceeq.tikz}}
\]
In particular, by means of such a wire, we can feed an output of a process back into it as an input:
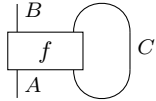
\[
\InputIfFileExists{Diagrams/trace_2.tikz}{}{\input{./figures/Diagrams/trace_2.tikz}}
\]
i.e.,~we create a feedback-loop.

Feedback-loops allow us to ask questions, such as how closely will some outgoing data match an~ingoing data. In particular, for the case of leaks where $L=A$, we can measure how closely the leaked data matches the original (while ignoring the output) via the following diagram:
 \ctikzfig{Diagrams/faithfulness}
However, what tends to be more useful, particularly in the case where $L\neq A$, is not asking precisely how well does the outgoing data match the ingoing, but how well does the outgoing data {encode} the ingoing data. For example, all of the information could be there just scrambled up or encoded in some other system type. We therefore want to consider maximising over potential {restoration} maps $r:L\to A$, where $r$ is taken to be causal. We call this notion the {quality} of a leak:
\[
\quality:=\textsf{Max}_r\left[\ %
\InputIfFileExists{Diagrams/quality.tikz}{}{\input{./figures/Diagrams/quality.tikz}}\ \right]
\]

If the structure of the numbers in a process theory is sufficiently rich, e.g.,~they are the real numbers or probabilities, one can moreover renormalise this quantity as follows:
\beq\label{eq:faithfulnessNORM}
\begin{tikzpicture}
	\begin{pgfonlayer}{nodelayer}
		\node [style=none] (0) at (2, 0.82500001) {$-$};
		\node [style={empty diagram}] (1) at (3.75, 0.82500001) {};
		\node [style=none] (2) at (0, 0.87500001) {$\quality$};
		\node [style=none] (3) at (2, 0) {};
	\end{pgfonlayer}
\end{tikzpicture}
}
\over
\InputIfFileExists{Diagrams/renormQuality1.tikz}{}{\input{./figures/Diagrams/renormQuality1.tikz}}
\eeq
where the circle indicates the feedback-loop applied to the identity. As a leak, the quality of broadcasting is one, since we have:
\[
\InputIfFileExists{Diagrams/faithfulness.tikz}{}{\input{./figures/Diagrams/faithfulness.tikz}}\ \ \stackrel{(\ref{eq:leakdisc2})}{=} \ \ %
\InputIfFileExists{Diagrams/circle.tikz}{}{\input{./figures/Diagrams/circle.tikz}}
\]
while for constant leaks, it is zero, since we have:
\[
\InputIfFileExists{Diagrams/faithfulnesscte1.tikz}{}{\input{./figures/Diagrams/faithfulnesscte1.tikz}}
\ \ =
\ \
\begin{tikzpicture}
	\begin{pgfonlayer}{nodelayer}
		\node [style=point] (0) at (0, -0.5) {$\rho$};
		\node [style=none] (1) at (0, 0.25) {};
		\node [style=upground] (2) at (0, 0.5) {};
	\end{pgfonlayer}
	\begin{pgfonlayer}{edgelayer}
		\draw (1.center) to (0);
	\end{pgfonlayer}
\end{tikzpicture}}
\ \ \stackrel{(\ref{eq:cteleak2})}{=} \ \
\emptydiag
\]

We therefore see that quantum theory is a minimally leaking theory as the renormalised quality for any leak is zero, whilst classical theory is maximal as every system has a leak with renormalised quality of one. In the next section, we consider how to increase the amount of leaking for a theory, providing a process-theoretic perspective on the quantum to classical transition.

\begin{example}\label{ex:continuous}
If a process theory admits sums (cf.~\cite{CKBook} or Appendix~\ref{mathsAppendix}), then set:
\[
\InputIfFileExists{Diagrams/c-leak.tikz}{}{\input{./figures/Diagrams/c-leak.tikz}}\ \ := \ \ c \ \ %
\InputIfFileExists{Diagrams/broadcast.tikz}{}{\input{./figures/Diagrams/broadcast.tikz}}\ \ + q \ \ %
\begin{tikzpicture}
	\begin{pgfonlayer}{nodelayer}
		\node [style=none] (0) at (-0.5, -1) {};
		\node [style=none] (1) at (-0.5, 1) {};
		\node [style=point] (2) at (0.5, 0) {$\rho$};
		\node [style=none] (3) at (0.5, 1) {};
	\end{pgfonlayer}
	\begin{pgfonlayer}{edgelayer}
		\draw [style=swap] (1.center) to (0.center);      
		\draw (3.center) to (2);
	\end{pgfonlayer}
\end{tikzpicture}}
\]
with $c + q = 1$. Now, quality in the form Equation (\ref{eq:faithfulnessNORM}) is $c$.
\end{example}

\section{A Representation for All Classical-Quantum Leaks}

We already characterised all quantum leaks as being constant leaks; we next characterise all classical leaks.

\begin{proposition}
All classical leaks are of the form:
\beq\label{eq:ClassicalLeak}
\InputIfFileExists{Diagrams/classicalLeak.tikz}{}{\input{./figures/Diagrams/classicalLeak.tikz}}
\eeq
where $l$ is any causal classical process.
\end{proposition}
\begin{proof}
First, let us define, using the non-causal ``cap'' of Example \ref{Ex:NonCausalCl}:
\[%
\InputIfFileExists{Diagrams/leakMap.tikz}{}{\input{./figures/Diagrams/leakMap.tikz}}\]
Despite the fact that this is defined using a non-causal process, the composite process $l$ is actually~causal:
\[%
\InputIfFileExists{Diagrams/leakMapIsCausal.tikz}{}{\input{./figures/Diagrams/leakMapIsCausal.tikz}}\]
We can then use the matrix representation of the leak (see Appendix~\ref{mathsAppendix}):
\[%
\InputIfFileExists{Diagrams/leakMatrix.tikz}{}{\input{./figures/Diagrams/leakMatrix.tikz}}\]
where $\Delta^{ij}_{k}\in \mathbb{R}^+$. The leak condition then implies that:
\[
\sum_j \Delta^{ij}_k=\delta^{i}_{k}
\]
and so:
\[
\Delta^{ij}_{k}=\Delta^{ij}_{k}\delta^{i}_{k}
\]
Then, we can check that Equation (\ref{eq:ClassicalLeak}) is indeed satisfied:
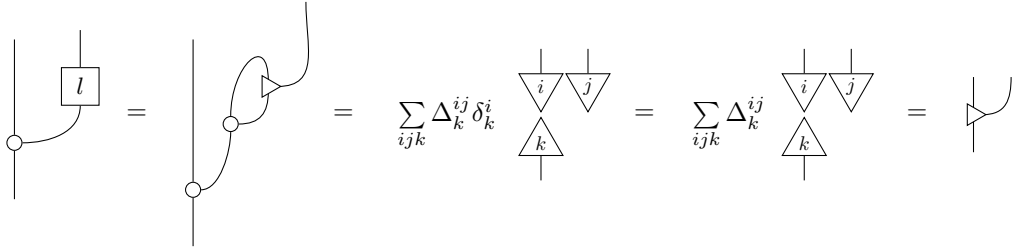
\[
\InputIfFileExists{Diagrams/ClassicalLeakFormProof.tikz}{}{\input{./figures/Diagrams/ClassicalLeakFormProof.tikz}}
\]
\end{proof}

We can now also characterise all leaks for composite classical-quantum systems:

\begin{proposition}
Denoting the classical system by a dotted line and the quantum system by a solid line; all~composite classical quantum systems have leaks of the form:
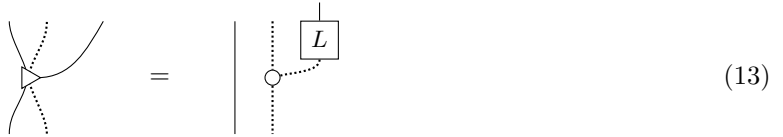
\beq%
\InputIfFileExists{Diagrams/classQuantLeakNew.tikz}{}{\input{./figures/Diagrams/classQuantLeakNew.tikz}}
\eeq
where $L$ is any causal process from classical to quantum systems.
\end{proposition}
\begin{proof}
Note that any composite leak defines a quantum leak as:
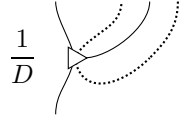
\[{1\over D}\ \ %
\InputIfFileExists{Diagrams/QuantumLeakFromComposite.tikz}{}{\input{./figures/Diagrams/QuantumLeakFromComposite.tikz}}\]
and therefore, as we know all quantum leaks separate:
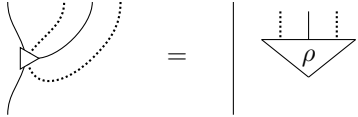
\[%
\InputIfFileExists{Diagrams/compLeakProof.tikz}{}{\input{./figures/Diagrams/compLeakProof.tikz}}\]
where $\rho$ defines a classical leak as:
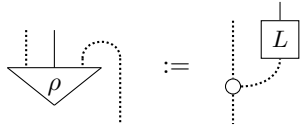
\[%
\InputIfFileExists{Diagrams/compLeakProof2.tikz}{}{\input{./figures/Diagrams/compLeakProof2.tikz}}\]
and so putting this together, we have:
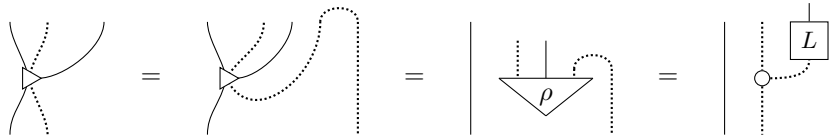
\[%
\InputIfFileExists{Diagrams/compLeakProof3.tikz}{}{\input{./figures/Diagrams/compLeakProof3.tikz}}\]
\end{proof}

The bottom line is that all of these leaks involve the copying leak as the fundamental ingredient. This is not all too surprising, since, as we showed in the previous section, it stands for maximal leakage. The processes $l$ and $L$ then play the role of reducing the leakage, with as the extremal cases $l$ and $L$ being constant, producing a constant leak.

\section{The Leak-Construction}

We now show how one can construct new process theories from old ones by introducing leaks. This is done by inserting particular processes of the old theory of the form (\ref{eq:decoher}) on all of the wires. The~processes (\ref{eq:plainleak}), to which we refer in the old theory as pre-leaks, then become leaks in the new theory. Hence, the leak construction turns pre-leaks into leaks.

\begin{thm}
\label{thm:leak-cons}
Given any process theory and for each system a {causal} process:
\beq\label{eq:plainleak}
\InputIfFileExists{Diagrams/plainpreleak.tikz}{}{\input{./figures/Diagrams/plainpreleak.tikz}}
\eeq
which is such that the following process is {idempotent}:
\beq\label{eq:decoher}
\InputIfFileExists{Diagrams/decoherPreLeak.tikz}{}{\input{./figures/Diagrams/decoherPreLeak.tikz}}
\eeq
and which are chosen {coherently} for composite systems:
\beq\label{eq:compositepreleak}
\InputIfFileExists{Diagrams/compositepreleakLHS.tikz}{}{\input{./figures/Diagrams/compositepreleakLHS.tikz}} \ \ :=\ \ %
\InputIfFileExists{Diagrams/compositepreleakRHS.tikz}{}{\input{./figures/Diagrams/compositepreleakRHS.tikz}}
\eeq
we can construct a new process theory in which each process (\ref{eq:plainleak}) is a leak for the system $A$. This construction goes as follows:
\begin{itemize}
\item systems stay the same;
\item one restricts processes to those of the form:
\beq\label{eq:leak-cons}
\InputIfFileExists{Diagrams/leak-cons.tikz}{}{\input{./figures/Diagrams/leak-cons.tikz}}
\eeq
\end{itemize}
\end{thm}
\begin{proof}
By causality of (\ref{eq:plainleak}):
\beq\label{preleakCausality}
\InputIfFileExists{Diagrams/causalPreLeak.tikz}{}{\input{./figures/Diagrams/causalPreLeak.tikz}}
\eeq
discarding is preserved by the leak-construction. Given the form Equation (\ref{eq:leak-cons}) of the processes in the theory and due to the idempotence of Equation (\ref{eq:decoher}), plain wires have taken the form Equation (\ref{eq:decoher}), so the defining equation of a leak Equation (\ref{eq:leakdisc1}) is satisfied. To consider the pre-leak in the new theory, we must apply the leak construction Equation (\ref{eq:leak-cons}), and using the condition for composites Equation (\ref{eq:compositepreleak}), we get the following process in the new theory:
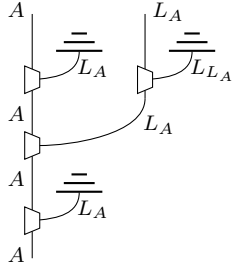
\[
\InputIfFileExists{Diagrams/preleakToLeak2.tikz}{}{\input{./figures/Diagrams/preleakToLeak2.tikz}}
\]
which is indeed a leak in the new theory:
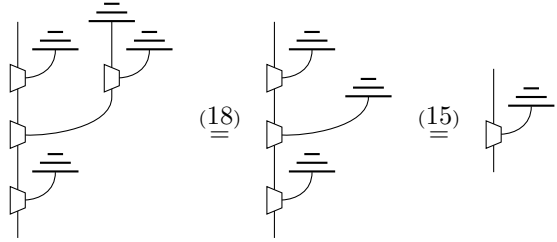
\[
\InputIfFileExists{Diagrams/preleakToLeak3.tikz}{}{\input{./figures/Diagrams/preleakToLeak3.tikz}}\ \
\stackrel{(\ref{preleakCausality})}{=} \ \
\InputIfFileExists{Diagrams/preleakToLeak4.tikz}{}{\input{./figures/Diagrams/preleakToLeak4.tikz}} \ \
\stackrel{(\ref{eq:decoher})}{=} \ \
\InputIfFileExists{Diagrams/preleakToLeak5.tikz}{}{\input{./figures/Diagrams/preleakToLeak5.tikz}}
\]
which is the form of a plain wire in the new theory,
and so, this construction does turn pre-leaks into leaks.
 It is moreover straightforward to see that we again obtain a process theory.
\end{proof}

Sometimes the leak-construction does nothing, in particular, when the pre-leaks are already leaks:

\begin{example}[Trivial]
A simple example of the leak construction is the one where the pre-leaks are taken to already be leaks, since then (\ref{eq:leak-cons}) will reduce to the processes $f$ themselves.
\end{example}

The main motivating example for this construction is of course the following:

\begin{example}[Decoherence]\label{ex:decoher-leak}
The leak construction for the pre-leak:
\[
\InputIfFileExists{Diagrams/broadcast.tikz}{}{\input{./figures/Diagrams/broadcast.tikz}}\ \ : {\cal B(H)}\to {\cal B(H\otimes H)}:: |i\rangle\langle i|\mapsto |i\rangle\langle i|\otimes |i\rangle\langle i|
\]
applied to the process theory of quantum processes (i.e., Example \ref{ex:quantumtheory}), we obtain classical probability theory (i.e.,~Example \ref{ex:classtheory}).
\end{example}

In the above construction, it is really the idempotents rather than the specific pre-leaks that determine the theory that is obtained. We can therefore have several different perspectives on the ``cause'' of this idempotent, by considering different pre-leaks from which it could be obtained. Firstly, we can always take the trivial case, where the pre-leak is just the idempotent itself, i.e., taking the leaked system as the empty system. There are however three alternate forms that always exist in quantum theory and that are more insightful.

\begin{example}
\label{ex:PreLeakCauses}
Firstly we can consider the {purification} $f$ of the idempotent, in the sense of \cite{Chiri1}:
\[%
\InputIfFileExists{Diagrams/preleakAsPurification.tikz}{}{\input{./figures/Diagrams/preleakAsPurification.tikz}}\]
This corresponds to the idea that information can never be fundamentally destroyed, only discarded, and so, we can see this leaking of information into some causally-separated system leading to decoherence. Another standard way to represent a general process is, via Stinespring dilation \cite{Stinespring}, as a reversible interaction with an environment:
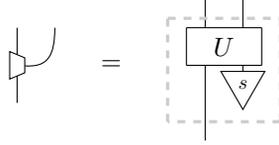
\[%
\InputIfFileExists{Diagrams/preleakAsUnitarification.tikz}{}{\input{./figures/Diagrams/preleakAsUnitarification.tikz}}\]
and so, we can equivalently view decoherence as arising due to a reversible interaction with some uncontrolled environment \cite{zurek2009quantum}. A final example, suggested to us by Rob Spekkens, is that the idempotent can be viewed as describing a system that lacks a reference frame \cite{bartlett2007reference}; the leaked system would then correspond to the reference system itself. This is the subject of ongoing work and is discussed in the Conclusion.
\end{example}

Example \ref{ex:decoher-leak} leaves open the question whether there are any theories that can be obtained from this leak construction in between classical and quantum theory. This question is solved in a forthcoming paper where the key result is the following theorem:

\begin{thm}
The leak construction applied to quantum processes (i.e., Example \ref{ex:quantumtheory}) gives all C*-algebras and C*-algebras only.
\end{thm}

Therefore, despite the weak structure of a leak, for the specific case of quantum theory, we obtain precisely the C*-algebras via the leak construction. This leads one to contemplate the view that the operational essence of (finite dimensional) C*-algebras is entirely captured by leaks and that the additional structure of C*-algebras is merely an artefact of the Hilbert space representation.

\begin{remark}
The leak-construction does not apply to Example \ref{ex:continuous}, since only for $c=0,1$, we have idempotence of~(\ref{eq:decoher}).
\end{remark}

\begin{remark}
For a process theory in which all systems are compositions $A^{\otimes n}$ of one atomic system $A$, it suffices to pick a single process (\ref{eq:plainleak}) for the system $A$ (where $L_A$ will be of the form $A^{\otimes n}$, since all other such processes arise then by coherence (\ref{eq:compositepreleak}).
\end{remark}

\begin{remark}
If a pre-leak with $L:= A$ is co-associative:
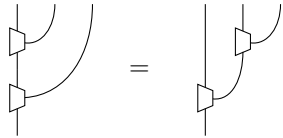
\ctikzfig{Diagrams/coassoc}
then the idempotence of (\ref{eq:decoher}) follows from causality of the pre-leak.
\end{remark}

\begin{remark}
The construction in Theorem \ref{thm:leak-cons}, when modified by not fixing a pre-leak for each type, but rather considering all pairs of a system and a corresponding pre-leak, is known as the Karoubi envelope, or Cauchy completion, or splitting of idempotents. More details on this can be found in \cite{CST}.
\end{remark}

\section{Process-Purity from Leaks} \label{Sec:ProcessPurity}

In this section, we consider how leaks relate to purity in process theories. The purity (or lack of purity) of a state is a fundamental concept in quantum theory and is equally important in most approaches to generalised physical theories. However, there is no reason to consider this as solely a~property for states, but should be considered for all processes in a theory. Indeed, lack of knowledge about a process, the noisiness of a channel and detection errors on a POVM-element all correspond to process-impurities. We will show that defining such a property for general theories, and classical theory in particular, requires leaks.

In Reference \cite{chiribella2016quantum}, Chiribella et al. introduce the notion of side-information; this can be thought of as information that is lost during a process that, in principle, could be possessed by some other agent. The use of this in cryptographic scenarios is clear, where the side-information can be thought of as being possessed by an eavesdropper attempting to influence or gain information about some cryptographic protocol. Diagrammatically, this side information is depicted as:
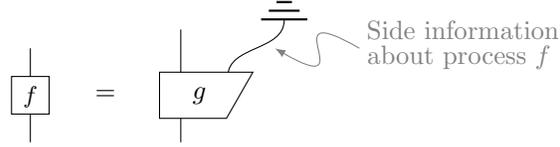
\[%
\InputIfFileExists{Diagrams/sideInformationProcess.tikz}{}{\input{./figures/Diagrams/sideInformationProcess.tikz}}\]

Lack of side-information for a process would imply that $g$ must separate such that the side-information is independent of the process $f$. Indeed, this must be the case for any such $g$,~i.e.:
\beq\label{eq:PurityOfProc2}
\begin{tikzpicture}
	\begin{pgfonlayer}{nodelayer}
		\node [style=box] (0) at (0, 0) {$f$};
		\node [style=none] (1) at (0, 1.25) {};
		\node [style=none] (2) at (0, -1.25) {};
	\end{pgfonlayer}
	\begin{pgfonlayer}{edgelayer}
		\draw (1.center) to (0); 
		\draw (0) to (2.center); 
	\end{pgfonlayer}
\end{tikzpicture}}\ \ =\ \ %
\InputIfFileExists{Diagrams/pure1.tikz}{}{\input{./figures/Diagrams/pure1.tikz}}\qquad \Longrightarrow\qquad %
\InputIfFileExists{Diagrams/pure2.tikz}{}{\input{./figures/Diagrams/pure2.tikz}}\ \ = \ \ %
\InputIfFileExists{Diagrams/pure4.tikz}{}{\input{./figures/Diagrams/pure4.tikz}}
\eeq
or in other words, all dilations of $f$ must separate. As mentioned in remark~\ref{Rem:QuantumPure}, the separability of dilations (cf. Proposition \ref{prop:reducedprocpure}) has been proposed as a definition of process-purity. Indeed for the case of quantum theory, this corresponds to the expected notion of purity, that is that the CPTP
 map must be Kraus rank $1$. remarkably, however, in the form of (\ref{eq:PurityOfProc2}), this definition does not extend to general processes of classical probability theory. In fact, nor does it do so for any theory that has broadcasting:

\begin{proposition}\label{prop:broadcastnogo}
If a non-trivial theory has broadcasting and one defines purity by means of (\ref{eq:PurityOfProc2}), then plain wires (i.e.,~identity processes) are not pure.
\end{proposition}
\begin{proof}
Assuming identities are pure and applying (\ref{eq:PurityOfProc2}) to the defining equation of a leak (\ref{eq:leakdisc1}), we obtain:
\beq\label{eq:broadcastnogo}
\InputIfFileExists{Diagrams/broadcast.tikz}{}{\input{./figures/Diagrams/broadcast.tikz}}\ \ = \ \ %
\begin{tikzpicture}
	\begin{pgfonlayer}{nodelayer}
		\node [style=none] (0) at (-0.5, -1) {};
		\node [style=none] (1) at (-0.5, 1) {};
		\node [style=point] (2) at (0.5, 0) {$\rho$};
		\node [style=none] (3) at (0.5, 1) {};
	\end{pgfonlayer}
	\begin{pgfonlayer}{edgelayer}
		\draw [style=swap] (1.center) to (0.center);      
		\draw (3.center) to (2);
	\end{pgfonlayer}
\end{tikzpicture}}
\eeq
that is, it is a constant leak. However, then, from the second defining equation of broadcasting, we obtain:
\[
\begin{tikzpicture}
	\begin{pgfonlayer}{nodelayer}
		\node [style=none] (0) at (0, -1.5) {};
		\node [style=none] (1) at (0, 1.5) {};
	\end{pgfonlayer}
	\begin{pgfonlayer}{edgelayer}
		\draw [style=swap] (1.center) to (0.center);
	\end{pgfonlayer}
\end{tikzpicture}}
\ \ \stackrel{(\ref{eq:leakdisc2})}{=} \ \
\InputIfFileExists{Diagrams/pure9.tikz}{}{\input{./figures/Diagrams/pure9.tikz}}
\ \ \stackrel{(\ref{eq:broadcastnogo})}{=} \ \
\begin{tikzpicture}
	\begin{pgfonlayer}{nodelayer}
		\node [style=point] (0) at (0, 0.75) {$\rho$};
		\node [style=none] (1) at (0, 1.75) {};
		\node [style=none] (2) at (0, -1.75) {};
		\node [style=upground] (3) at (0, -0.75) {};
	\end{pgfonlayer}
	\begin{pgfonlayer}{edgelayer}
		\draw (1.center) to (0);
		\draw [style=swap] (3) to (2.center); 
	\end{pgfonlayer}
\end{tikzpicture}}
\]
that is, each plain wire is a constant process, and hence, the theory is trivial, since as a consequence, all processes must then be constant since for (causal) processes, we have:
\[
\begin{tikzpicture}
	\begin{pgfonlayer}{nodelayer}
		\node [style=box] (0) at (0, 0) {$f$};
		\node [style=none] (1) at (0, 1.25) {};
		\node [style=none] (2) at (0, -1.25) {};
	\end{pgfonlayer}
	\begin{pgfonlayer}{edgelayer}
		\draw (1.center) to (0); 
		\draw (0) to (2.center); 
	\end{pgfonlayer}
\end{tikzpicture}}\ \ =\ \ %
\begin{tikzpicture}
	\begin{pgfonlayer}{nodelayer}
		\node [style=box] (0) at (0, -1.5) {$f$};
		\node [style=none] (1) at (0, 2.75) {};
		\node [style=none] (2) at (0, -2.75) {};
	\end{pgfonlayer}
	\begin{pgfonlayer}{edgelayer}
		\draw (1.center) to (0);
		\draw (0) to (2.center);
	\end{pgfonlayer}
\end{tikzpicture}}\ \ =\ \ %
\InputIfFileExists{Diagrams/f-box-disc.tikz}{}{\input{./figures/Diagrams/f-box-disc.tikz}}\ \ =\ \ %
\begin{tikzpicture}
	\begin{pgfonlayer}{nodelayer}
		\node [style=point] (0) at (0, 0.75) {$\rho$};
		\node [style=none] (1) at (0, 1.75) {};
		\node [style=none] (2) at (0, -1.75) {};
		\node [style=upground] (3) at (0, -0.75) {};
	\end{pgfonlayer}
	\begin{pgfonlayer}{edgelayer}
		\draw (1.center) to (0);
		\draw [style=swap] (3) to (2.center); 
	\end{pgfonlayer}
\end{tikzpicture}}
\]
Hence, in a non-trivial theory with broadcasting, identities cannot be pure in the sense of (\ref{eq:PurityOfProc2}).
\end{proof}

From the first part of this proof, namely that this definition of purity implies that leaks must be constant, it follows that this issue arises in any theories with non-constant leaks. We can think of this as the fact that, if a system has a leak, then there is irreducible side-information contained within the system itself:
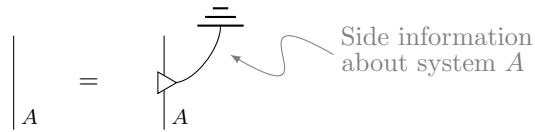
\[
\InputIfFileExists{Diagrams/sideInformationSystem.tikz}{}{\input{./figures/Diagrams/sideInformationSystem.tikz}}
\]

Fortunately, leaks also allow us to fix this problem. Firstly, let us suppose that a theory has leaks and also has a pure process $f$. Then, clearly, the following is a dilation of $f$:
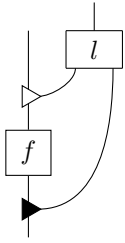
\[
\InputIfFileExists{Diagrams/leakBeforeAfter.tikz}{}{\input{./figures/Diagrams/leakBeforeAfter.tikz}}
\]
where $l$ is causal. One may therefore consider explicitly bringing leaks into play in the definition of purity. A first step in this direction is to weaken (\ref{eq:PurityOfProc2}) as follows:
\beq\label{def:EI}
\begin{tikzpicture}
	\begin{pgfonlayer}{nodelayer}
		\node [style=box] (0) at (0, 0) {$f$};
		\node [style=none] (1) at (0, 1.25) {};
		\node [style=none] (2) at (0, -1.25) {};
	\end{pgfonlayer}
	\begin{pgfonlayer}{edgelayer}
		\draw (1.center) to (0); 
		\draw (0) to (2.center); 
	\end{pgfonlayer}
\end{tikzpicture}}\ \ =\ \ %
\InputIfFileExists{Diagrams/pure1.tikz}{}{\input{./figures/Diagrams/pure1.tikz}}\qquad \Longrightarrow\qquad \exists\ \ %
\InputIfFileExists{Diagrams/plainleak.tikz}{}{\input{./figures/Diagrams/plainleak.tikz}},\ \ %
\InputIfFileExists{Diagrams/j2.tikz}{}{\input{./figures/Diagrams/j2.tikz}}\ \&\ \ %
\InputIfFileExists{Diagrams/causalL.tikz}{}{\input{./figures/Diagrams/causalL.tikz}}\ \ : \ \ %
\InputIfFileExists{Diagrams/pure2.tikz}{}{\input{./figures/Diagrams/pure2.tikz}}\ \ = \ \ %
\InputIfFileExists{Diagrams/leakBeforeAfter.tikz}{}{\input{./figures/Diagrams/leakBeforeAfter.tikz}}\eeq

 However, now, we have the opposite problem: all classical processes, including all states, are pure!
 (See Appendix~\ref{proof:ClassicalDilations}.) It is clear that we are missing a constraint. The original idea was that for a~process to be pure, it should have no side-information that some eavesdropper could take advantage of. However, we have shown that for some systems, there is irreducible side-information represented by leakage. Therefore, to ensure that the eavesdropper cannot gain information or influence the process, we must demand that the process does not interact with this irreducible side-information, such that leaking before or after is equivalent:

\beq \label{def:LP} \forall\ \ %
\InputIfFileExists{Diagrams/plainleak.tikz}{}{\input{./figures/Diagrams/plainleak.tikz}}\ \ \exists\ \ %
\InputIfFileExists{Diagrams/j2.tikz}{}{\input{./figures/Diagrams/j2.tikz}} \text{ and }\ \forall\ \ %
\InputIfFileExists{Diagrams/j2.tikz}{}{\input{./figures/Diagrams/j2.tikz}}\ \ \exists \ \ %
\InputIfFileExists{Diagrams/plainleak.tikz}{}{\input{./figures/Diagrams/plainleak.tikz}} \text{ such that }\ \ %
\InputIfFileExists{Diagrams/pure3.tikz}{}{\input{./figures/Diagrams/pure3.tikz}}\ \ = \ \ %
\InputIfFileExists{Diagrams/j1.tikz}{}{\input{./figures/Diagrams/j1.tikz}} \eeq

Hence, we propose the following definition of process-purity, which packages these two conditions, (\ref{def:EI}) and (\ref{def:LP}), into a neat form:

\begin{definition} $f$ is pure if and only if:
\label{def:pureeq}
\beq
\begin{tikzpicture}
	\begin{pgfonlayer}{nodelayer}
		\node [style=box] (0) at (0, 0) {$f$};
		\node [style=none] (1) at (0, 1.25) {};
		\node [style=none] (2) at (0, -1.25) {};
	\end{pgfonlayer}
	\begin{pgfonlayer}{edgelayer}
		\draw (1.center) to (0); 
		\draw (0) to (2.center); 
	\end{pgfonlayer}
\end{tikzpicture}}\ \ =\ \ %
\InputIfFileExists{Diagrams/pure1.tikz}{}{\input{./figures/Diagrams/pure1.tikz}}\quad \Longrightarrow\quad \exists\ \ %
\InputIfFileExists{Diagrams/plainleak.tikz}{}{\input{./figures/Diagrams/plainleak.tikz}}\ \& \ \ %
\InputIfFileExists{Diagrams/j2.tikz}{}{\input{./figures/Diagrams/j2.tikz}} \ \ : \ \ %
\InputIfFileExists{Diagrams/pure2.tikz}{}{\input{./figures/Diagrams/pure2.tikz}}\ \ = \ \ %
\InputIfFileExists{Diagrams/pure3.tikz}{}{\input{./figures/Diagrams/pure3.tikz}}\ \ = \ \ %
\InputIfFileExists{Diagrams/leakAfter.tikz}{}{\input{./figures/Diagrams/leakAfter.tikz}}\eeq
\end{definition}

This ensures that the only side-information is this irreducible kind, i.e., system leakage, and~moreover, that pure processes do not interact with this irreducible side information. To further motivate this definition, we will show that it provides a sensible definition for quantum, classical and composite systems. However, first, note that for states, this definition reduces to:

\begin{example}
A state $\psi$ is pure if we have:
\[
\begin{tikzpicture}
	\begin{pgfonlayer}{nodelayer}
		\node [style=point] (0) at (0, -0.25) {$\psi$};
		\node [style=none] (1) at (0, 0.75) {};
	\end{pgfonlayer}
	\begin{pgfonlayer}{edgelayer}
		\draw (1.center) to (0);
	\end{pgfonlayer}
\end{tikzpicture}}\ \ =\ \ %
\InputIfFileExists{Diagrams/pure5.tikz}{}{\input{./figures/Diagrams/pure5.tikz}}\qquad \Longrightarrow\qquad %
\InputIfFileExists{Diagrams/pure6.tikz}{}{\input{./figures/Diagrams/pure6.tikz}}\ \ = \ \ %
\begin{tikzpicture}
	\begin{pgfonlayer}{nodelayer}
		\node [style=point] (0) at (-0.75, -0.25) {$\psi$};
		\node [style=none] (1) at (-0.75, 0.75) {};
		\node [style=none] (2) at (0.75, 0.75) {};
		\node [style=point] (3) at (0.75, -0.25) {$\rho$};
	\end{pgfonlayer}
	\begin{pgfonlayer}{edgelayer}
		\draw (1.center) to (0);
		\draw (2.center) to (3);
	\end{pgfonlayer}
\end{tikzpicture}}
\]
\end{example}

This is the same as the original definition, and so, we see that it is only for general processes that this new definition is necessary. Similarly, in the case of quantum theory, it is only the first condition that provides a non-trivial constraint:
\begin{example}[quantum purity]
As for quantum theory, the only leaks are constant leaks, Condition (\ref{def:EI}) in Definition \ref{def:pureeq} reduces to (\ref{eq:PurityOfProc2}), while Condition (\ref{def:LP}) becomes trivial.
 \end{example}

Whilst, in the classical case, as we have mentioned above, (\ref{def:EI}) is satisfied by all classical processes, and so, it is only (\ref{def:LP}) that needs to be considered:
 \begin{example}[classical purity]\label{ex:classicalPurity}
All pure classical processes, between an $n$ and $m$ state system, are of the form:
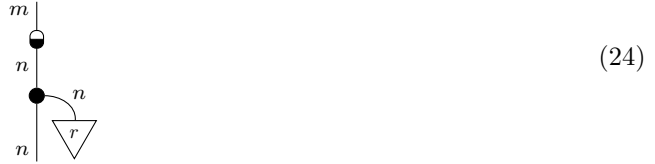
\beq
\InputIfFileExists{Diagrams/pureClassicalNew.tikz}{}{\input{./figures/Diagrams/pureClassicalNew.tikz}}
\eeq
where we can define the `upside-down broadcasting map' by:
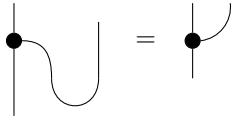
\[%
\InputIfFileExists{Diagrams/classicalMatching.tikz}{}{\input{./figures/Diagrams/classicalMatching.tikz}}\]
and the black/white dot is any process that satisfies:
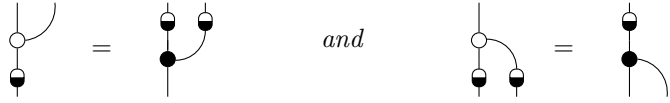
\[%
\InputIfFileExists{Diagrams/blackWhite.tikz}{}{\input{./figures/Diagrams/blackWhite.tikz}} \qquad\qquad \text{and} \qquad\qquad %
\InputIfFileExists{Diagrams/blackWhite2.tikz}{}{\input{./figures/Diagrams/blackWhite2.tikz}}\]
\end{example}

\proof
We prove here that pure classical processes must be of this form and leave the proof that any process of this form is pure to Appendix~\ref{proof:PureProcesses}.

First consider the condition:
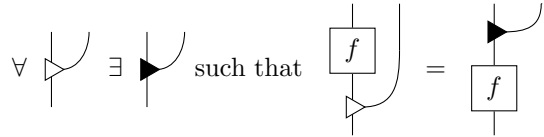
\[ \forall\ \ %
\InputIfFileExists{Diagrams/plainleak.tikz}{}{\input{./figures/Diagrams/plainleak.tikz}}\ \ \exists\ \ %
\InputIfFileExists{Diagrams/j2.tikz}{}{\input{./figures/Diagrams/j2.tikz}} \text{ such that }\ \ %
\InputIfFileExists{Diagrams/pure3.tikz}{}{\input{./figures/Diagrams/pure3.tikz}}\ \ = \ \ %
\InputIfFileExists{Diagrams/j1.tikz}{}{\input{./figures/Diagrams/j1.tikz}} \]
for the special case where:
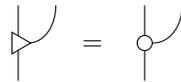
\[%
\InputIfFileExists{Diagrams/plainleak.tikz}{}{\input{./figures/Diagrams/plainleak.tikz}}\ \ = \ \ \ %
\InputIfFileExists{Diagrams/broadcast.tikz}{}{\input{./figures/Diagrams/broadcast.tikz}}\]
and using the standard form for classical leaks to write:
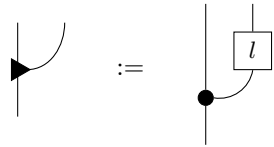
\[%
\InputIfFileExists{Diagrams/CPP1.tikz}{}{\input{./figures/Diagrams/CPP1.tikz}}\]
Then, we can show that:
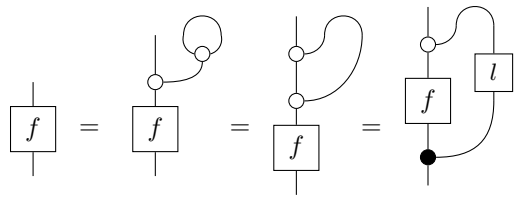
\[%
\InputIfFileExists{Diagrams/CPP2.tikz}{}{\input{./figures/Diagrams/CPP2.tikz}}\]
This implies that, for all $i$ and $j$:
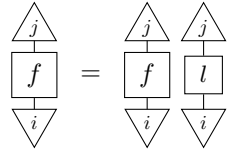
\[%
\InputIfFileExists{Diagrams/CPP3.tikz}{}{\input{./figures/Diagrams/CPP3.tikz}}\]
so, for each $i$ and $j$:
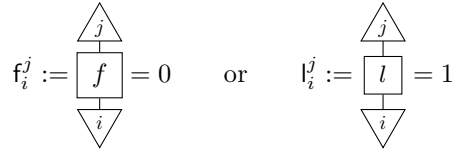
\[\mathsf{f}_i^j:=%
\InputIfFileExists{Diagrams/CPP4.tikz}{}{\input{./figures/Diagrams/CPP4.tikz}}=0\qquad \text{or} \qquad \mathsf{l}_i^j:=%
\InputIfFileExists{Diagrams/CPP5.tikz}{}{\input{./figures/Diagrams/CPP5.tikz}}=1\]
Causality of $l$ then implies that, for each $j$, there can only be a single $i$ where $\mathsf{l}_i^j=1$, and so, for all other $i$, we must have $\mathsf{f}_i^j=0$. This means that in each row of $\mathsf{f}_i^j$, there is at most a single non-zero element.

We can run through this argument in the opposite direction using the condition:
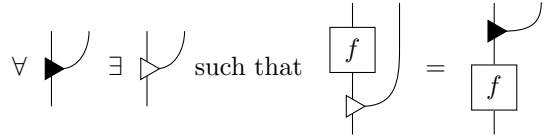
\[ \forall\ \ %
\InputIfFileExists{Diagrams/j2.tikz}{}{\input{./figures/Diagrams/j2.tikz}}\ \ \exists \ \ %
\InputIfFileExists{Diagrams/plainleak.tikz}{}{\input{./figures/Diagrams/plainleak.tikz}} \text{ such that }\ \ %
\InputIfFileExists{Diagrams/pure3.tikz}{}{\input{./figures/Diagrams/pure3.tikz}}\ \ = \ \ %
\InputIfFileExists{Diagrams/j1.tikz}{}{\input{./figures/Diagrams/j1.tikz}} \]
which shows that $\mathsf{f}_i^j$ can have at most a single non-zero element in each column. This is precisely what is enforced by the black/white dot in the above form; the value of the non-zero elements is then determined by the state $r$. Hence, we can write $f$ in the desired form.
\endproof

\begin{example}
If we consider purity for {causal} classical processes, then we find that the pure processes are those that are reversible (i.e., are isometries).
\end{example}
\proof
The definition of purity, and the standard form for classical leaks, requires that:
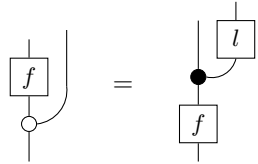
\[%
\InputIfFileExists{Diagrams/pureCausalClassical1.tikz}{}{\input{./figures/Diagrams/pureCausalClassical1.tikz}}\]
and so, we have:
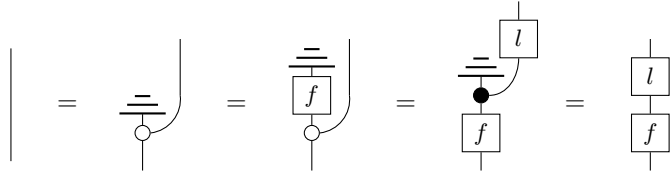
\[%
\InputIfFileExists{Diagrams/pureCausalClassical2.tikz}{}{\input{./figures/Diagrams/pureCausalClassical2.tikz}}\]
Therefore, $f$ is reversible in the sense that it has a left-inverse, i.e., $l$.
\endproof

Finally, we consider the composite case, where the conjunction of (\ref{def:EI}) and (\ref{def:LP}) is necessary:

\begin{example}[Composite classical-quantum purity]
Pure processes are:
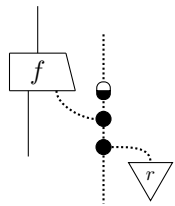
\beq%
\InputIfFileExists{Diagrams/classQuantPureProcNew.tikz}{}{\input{./figures/Diagrams/classQuantPureProcNew.tikz}}\eeq
where we denote the classical system with a dotted line, and:
\beq%
\InputIfFileExists{Diagrams/classQuantPureProc2.tikz}{}{\input{./figures/Diagrams/classQuantPureProc2.tikz}}\eeq
is pure for all $i$.
\end{example}
\proof
Again, we prove the interesting direction here that pure processes on composite systems must be of this form and, again, leave the other direction to Appendix \ref{proof:PureProcesses}.

Note that a generic process can be written as:
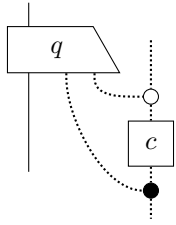
\[%
\InputIfFileExists{Diagrams/QCPP1.tikz}{}{\input{./figures/Diagrams/QCPP1.tikz}}\]
An (almost) identical argument to the classical case shows that if this is pure, it can be written as:
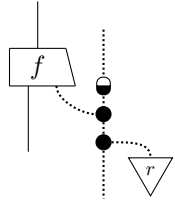
\[%
\InputIfFileExists{Diagrams/classQuantPureProcNew.tikz}{}{\input{./figures/Diagrams/classQuantPureProcNew.tikz}}\]
We therefore move on to considering the other part of the definition of purity, that is that any dilation can be written as a leak; that means that any dilation of this process can be written as:
\[%
\InputIfFileExists{Diagrams/QCPP4.tikz}{}{\input{./figures/Diagrams/QCPP4.tikz}}\ \ =\ \sum_i\ \ %
\InputIfFileExists{Diagrams/QCPP3.tikz}{}{\input{./figures/Diagrams/QCPP3.tikz}}\]
Now, note that any collection of dilations of the processes:
\[%
\InputIfFileExists{Diagrams/classQuantPureProc2.tikz}{}{\input{./figures/Diagrams/classQuantPureProc2.tikz}}\ \ := \ \ %
\begin{tikzpicture}
	\begin{pgfonlayer}{nodelayer}
		\node [style=box] (0) at (0, -0) {$f_i$};
		\node [style=none] (1) at (0, 1) {};
		\node [style=none] (2) at (0, -0.9999998) {};
	\end{pgfonlayer}
	\begin{pgfonlayer}{edgelayer}
		\draw (1.center) to (0);
		\draw (0) to (2.center);
	\end{pgfonlayer}
\end{tikzpicture}
} \ \ =\ \ %
\InputIfFileExists{Diagrams/QCPP5.tikz}{}{\input{./figures/Diagrams/QCPP5.tikz}}\]
defines a dilation of the whole process, which must be able to be written as a leak:
\[\sum_i\ \ %
\InputIfFileExists{Diagrams/QCPP6.tikz}{}{\input{./figures/Diagrams/QCPP6.tikz}}\ \ =\ \sum_i\ \ %
\InputIfFileExists{Diagrams/QCPP3.tikz}{}{\input{./figures/Diagrams/QCPP3.tikz}}\]
Therefore, each $g_i$ must separate, and hence, the $f_i$ are each pure quantum processes.
\endproof

An immediate consequence of this is the following.
\begin{proposition}
The pure quantum to classical or classical to quantum maps are separable.
\end{proposition}
\proof
First note that,
\[ %
\InputIfFileExists{Diagrams/CtoQpureMaps.tikz}{}{\input{./figures/Diagrams/CtoQpureMaps.tikz}} \]
Then, using the above result regarding pure maps for composite quantum classical systems, we have,
\[ %
\InputIfFileExists{Diagrams/CtoQpureMaps2.tikz}{}{\input{./figures/Diagrams/CtoQpureMaps2.tikz}}\]
Similarly, we obtain separability for pure quantum to classical maps.
\endproof

This means that there is no {pure} way to transform between classical and quantum information.

\section{Conclusions}

In this paper, we introduce the concept of {leaks} to generalised process theories. The definition of which can be thought of as a ``one-way'' broadcasting map. These prove to be very useful for understanding various aspects of quantum theory from a physically well-motivated perspective. In particular, we show:

\begin{itemize}
\item that quantum theory is a leak-free theory, whilst classical theory is maximally leaking, giving a~clear separation between the theories for which quantum theory is optimal.
\item how to construct sub-theories via a ``leak construction'', which can be thought of as the sub-theories that can be obtained from a dynamical decoherence mechanism. For quantum theory, we can obtain classical theory, composite quantum classical theory and, generally, finite dimensional C*-algebras from this construction \cite{CST}.
\item a characterisation of the leaks and pure processes for quantum, classical and composite systems; in particular, we demonstrate that there is no pure way to transform quantum systems into classical systems or vice versa.
\item that leaks are essential to defining purity of processes; we therefore introduce a novel definition of purity of processes, which makes sense both for quantum theory and for classical theory.
\end{itemize}

\subsection*{Future Work}
In this paper, we have shown how classical theory emerges from quantum theory due to the leak construction, providing a process-theoretic perspective on why the world on a large scale appears to us to be classical. It is natural to ask: Is there some deeper theory of nature than quantum theory from which quantum theory emerges in an analogous way? This is the subject of a forthcoming paper \cite{hyperdecoherence}. A second, related question, would be to ask: What does it imply about a theory if it can obtain classical theory via a leak construction; is the ability for this to happen in quantum theory special or is this a generic feature of general theories?

We have also shown that quantum theory is minimally leaking and classical theory maximally; moreover, if we start from a process theory describing finite dimensional C*-algebras, then quantum theory is singled out as the unique minimally-leaking theory. Can this idea lead to a complete reconstruction of quantum theory \cite{reconstruction}?

As mentioned in Example~\ref{ex:PreLeakCauses}, one interpretation of the leak construction is as a way to represent systems for which there is a missing reference frame, that is we can write the pre-leak as (\cite{bartlett2007reference}, Section~IVB):
\[%
\InputIfFileExists{Diagrams/preleakAsReferenceFrame.tikz}{}{\input{./figures/Diagrams/preleakAsReferenceFrame.tikz}}\]
where $G$ is a group associated with a reference frame for a particular degree of freedom, $U_g$ is the representation of $G$ on the system of interest and $g$ the state of the reference system. Note, however, that making sense of this integral for general symmetry groups requires the reference be an infinite dimensional quantum system and so is beyond the scope of this paper. One could replace, at least for compact groups, the integral by a finite convex mixture (using the results of \cite{Chiri1}, Corollary 33 from Caratheodory's theorem), for which the resulting idempotent would be the same. This can be thought of as there only being a finite set of possible orientations for the reference frame. However, a comprehensive understanding of the connections here would require consideration of the infinite dimensional case. Moreover, we know that in the finite dimensional case, the leak construction leads to C*-algebraic systems only; however, it remains an interesting open question as to what the leak construction leads to for infinite dimensional systems.

\vspace{6pt}
\section*{Acknowledgments}We thank Aleks Kissinger, Dan Marsden, Rob Spekkens and Sean Tull for useful feedback. John~Selby was supported by the EPSRC 
 through the Controlled Quantum Dynamics Centre for Doctoral Training, and Bob Coecke is supported by the U.S. Air Force Office of Scientific Research.




\appendix

\section{Mathematical Tools for Proofs}\label{mathsAppendix}

\begin{definition}
Sums can be defined by the fact that they distribute over diagrams, that is:
\[%
\InputIfFileExists{Diagrams/sumsDistribute.tikz}{}{\input{./figures/Diagrams/sumsDistribute.tikz}}\]
\end{definition}

In particular, in classical probability theory, we can take sums of diagrams where the sum is the standard sum of matrices. In fact, this provides us with a matrix calculus for our diagrams. In~particular, we have a basis and co-basis for each system, denoted:
\[\left\{%
\begin{tikzpicture}
	\begin{pgfonlayer}{nodelayer}
		\node [style=point] (0) at (0, -0) {$i$};
		\node [style=none] (1) at (0, 0.9999999) {};
	\end{pgfonlayer}
	\begin{pgfonlayer}{edgelayer}
		\draw (1.center) to (0);
	\end{pgfonlayer}
\end{tikzpicture}
}\right\}_{i=1}^n\quad \text{and}\quad\left\{%
\begin{tikzpicture}
	\begin{pgfonlayer}{nodelayer}
		\node [style=copoint] (0) at (0, -0) {$j$};
		\node [style=none] (1) at (0, -0.9999999) {};
	\end{pgfonlayer}
	\begin{pgfonlayer}{edgelayer}
		\draw (1.center) to (0);
	\end{pgfonlayer}
\end{tikzpicture}
}\right\}_{j=1}^n\]
respectively, such that they are orthonormal:
\[%
\begin{tikzpicture}
	\begin{pgfonlayer}{nodelayer}
		\node [style=copoint] (0) at (0, 0.4999999) {$j$};
		\node [style=point] (1) at (0, -0.5000003) {$i$};
	\end{pgfonlayer}
	\begin{pgfonlayer}{edgelayer}
		\draw (1) to (0);
	\end{pgfonlayer}
\end{tikzpicture}
}\ \ =\ \ \delta_{ij}\]
Then, this provides a decomposition of the identity
\[%
\InputIfFileExists{Diagrams/decompositionIdentity.tikz}{}{\input{./figures/Diagrams/decompositionIdentity.tikz}}\]
which allows us to write any process as:
\[%
\InputIfFileExists{Diagrams/matrixRepresentation.tikz}{}{\input{./figures/Diagrams/matrixRepresentation.tikz}}\]
where it is simple to check that sequential composition then coincides with a matrix multiplication, parallel composition with the matrix tensor product and diagrammatic sum with the sum of matrices.

\begin{definition}
For each classical system type, we have a family of {spiders} diagrammatically defined by, firstly:
\[%
\InputIfFileExists{Diagrams/spiderFusion.tikz}{}{\input{./figures/Diagrams/spiderFusion.tikz}}\]
and secondly, that the symmetries of the representation as spiders are respected.
Alternately, spiders can be defined via the matrix representation as:
\[%
\InputIfFileExists{Diagrams/spiderMatrix.tikz}{}{\input{./figures/Diagrams/spiderMatrix.tikz}}\]
\end{definition}
This family of maps is particularly important as, for classical theory at least, they allow us to define various concepts that we have used throughout the paper in a unified way.
Firstly, the~broadcasting map can now be seen as just an example spider with one input and two outputs, but~moreover, we have:
\[%
\InputIfFileExists{Diagrams/cupSpider.tikz}{}{\input{./figures/Diagrams/cupSpider.tikz}}\qquad\qquad%
\InputIfFileExists{Diagrams/capSpider.tikz}{}{\input{./figures/Diagrams/capSpider.tikz}}\qquad\qquad%
\InputIfFileExists{Diagrams/discardSpider.tikz}{}{\input{./figures/Diagrams/discardSpider.tikz}}\]
The feedback-loop we introduced can also now be interpreted as the composite of two spiders:
\[%
\InputIfFileExists{Diagrams/feedbackloopSpider.tikz}{}{\input{./figures/Diagrams/feedbackloopSpider.tikz}}\]

We moreover want to consider a way to join spiders of different dimensionality (denoted by using a different colour), which is exactly what the black/white dots achieve.
\begin{definition}Diagrammatically, the black/white dots are any process satisfying:
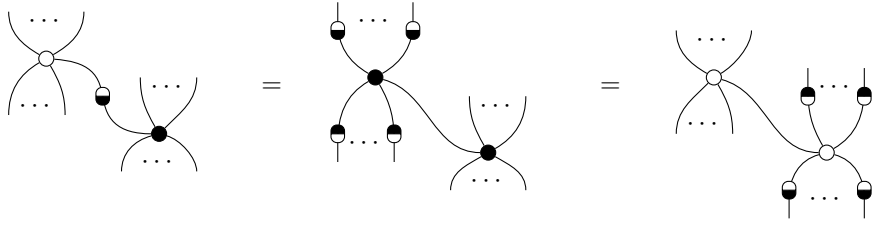
\[%
\InputIfFileExists{Diagrams/spiderLink.tikz}{}{\input{./figures/Diagrams/spiderLink.tikz}}\]
which is equivalent to how they were introduced in Example~\ref{ex:classicalPurity}. Alternatively, their matrix representation is:
\beq%
\InputIfFileExists{Diagrams/wbSpider.tikz}{}{\input{./figures/Diagrams/wbSpider.tikz}}\eeq
requiring that $l\leq \mathsf{Min}(n,m)$ and $\pi_i$ are arbitrary permutations of the basis elements. These are then just matrices with elements $\{0,1\}$ with at most a single one in each row and column.
\end{definition}

\section{Dilations of Classical Processes}\label{proof:ClassicalDilations}
 Any dilation of a classical process $f$ can be written as:
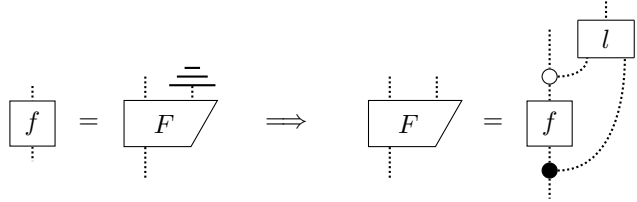
\[%
\InputIfFileExists{Diagrams/classicalDilation.tikz}{}{\input{./figures/Diagrams/classicalDilation.tikz}}\]
to check this define $l$ by its matrix elements as:
\[\mathsf{l}_{ik}^j := \left\{ \begin{array}{ccc} 1 & & \text{if}\quad \mathsf{f}_i^k=0 \\ & \\ \frac{\mathsf{F}_i^{kj}}{\mathsf{f}_i^k} & & \text{otherwise} \end{array}\right. \]
and then, it is simple to check this satisfies the above equation and, moreover, is causal.

\section{Pure Quantum-Classical Composite Processes}\label{proof:PureProcesses}
We need to prove that our definition of purity, i.e.,~Conditions (\ref{def:LP}) and (\ref{def:EI}), is satisfied by any process of the form:
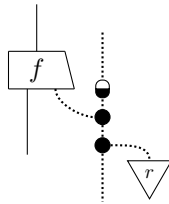
\[%
\InputIfFileExists{Diagrams/classQuantPureProcNew.tikz}{}{\input{./figures/Diagrams/classQuantPureProcNew.tikz}}\]
where:
\[%
\begin{tikzpicture}
	\begin{pgfonlayer}{nodelayer}
		\node [style=box] (0) at (0, -0) {$f_i$};
		\node [style=none] (1) at (0, 1) {};
		\node [style=none] (2) at (0, -0.9999998) {};
	\end{pgfonlayer}
	\begin{pgfonlayer}{edgelayer}
		\draw [style=cWire] (0) to (1.center);
		\draw [style=cWire] (0) to (2.center);
	\end{pgfonlayer}
\end{tikzpicture}
}\ \ :=\ \ %
\InputIfFileExists{Diagrams/classQuantPureProc2.tikz}{}{\input{./figures/Diagrams/classQuantPureProc2.tikz}}\]
is pure for all $i$.

That (\ref{def:LP}) is satisfied is a straightforward proof once the following observation, easily verified by a straightforward calculation, is made:
	\[ \forall\ \ %
\begin{tikzpicture}
	\begin{pgfonlayer}{nodelayer}
		\node [style=small box] (0) at (0, -0) {$l$};
		\node [style=none] (1) at (0, -1) {};
		\node [style=none] (2) at (0, 1) {};
	\end{pgfonlayer}
	\begin{pgfonlayer}{edgelayer}
		\draw [style=cWire] (0) to (1.center);
		\draw [style=cWire] (2.center) to (0);
	\end{pgfonlayer}
\end{tikzpicture}}\ \ \exists \ \ %
\begin{tikzpicture}
	\begin{pgfonlayer}{nodelayer}
		\node [style=small box] (0) at (0, -0) {$\tilde{l}$};
		\node [style=none] (1) at (0, 1) {};
		\node [style=none] (2) at (0, -1) {};
	\end{pgfonlayer}
	\begin{pgfonlayer}{edgelayer}
		\draw [style=cWire] (1.center) to (0.center);
		\draw [style=cWire] (0.center) to (2.center);
	\end{pgfonlayer}
\end{tikzpicture}}\] such that $l$ and $\tilde{l}$ are both causal and: \[ %
\InputIfFileExists{Diagrams/causalExtension3.tikz}{}{\input{./figures/Diagrams/causalExtension3.tikz}} \]
 To check this, simply define $\tilde{l}$ as:
 \[ %
\InputIfFileExists{Diagrams/causalExtension4.tikz}{}{\input{./figures/Diagrams/causalExtension4.tikz}} \]
 where $J=\mathsf{Ker}\left[%
\begin{tikzpicture}
	\begin{pgfonlayer}{nodelayer}
		\node [style=bwSpider] (0) at (0, -0) {};
		\node [style=none] (1) at (0, 0.5000001) {};
		\node [style=none] (2) at (0, -0.5000001) {};
	\end{pgfonlayer}
	\begin{pgfonlayer}{edgelayer}
		\draw [style=cWire](0) to (1.center);
		\draw [style=cWire](0) to (2.center);
	\end{pgfonlayer}
\end{tikzpicture}
}\right]$.

That (\ref{def:EI}) is satisfied is also simple if, using the purity of the $f_i$, we can write the dilation as:
\[\sum_i \ \ %
\InputIfFileExists{Diagrams/genericDilation.tikz}{}{\input{./figures/Diagrams/genericDilation.tikz}}\]
and note that this can be written as a leak by defining:
\[%
\InputIfFileExists{Diagrams/leaksFromDilations2.tikz}{}{\input{./figures/Diagrams/leaksFromDilations2.tikz}}\qquad \ \text{where} \qquad \ %
\InputIfFileExists{Diagrams/leakFromDilations.tikz}{}{\input{./figures/Diagrams/leakFromDilations.tikz}}\]

\bibliographystyle{plain}
\bibliography{main}

\end{document}